\documentclass[lettersize,journal]{IEEEtran}
	\usepackage{amsmath,amsfonts}
	\usepackage{algorithmic}
	\usepackage{algorithm}
	\usepackage{array}
	\usepackage[caption=false,font=normalsize,labelfont=sf,textfont=sf]{subfig}
	\usepackage{textcomp}
	\usepackage{stfloats}
	\usepackage{url}
	\usepackage{verbatim}
	\usepackage{graphicx}
	\usepackage{bm}
	\usepackage{cite}
	\usepackage[numbers]{natbib}
	\usepackage{amsthm}  

    \usepackage{array}
    \usepackage{makecell}
    \usepackage{amssymb} 
    \usepackage{url}
   \usepackage{titlesec}  
   \titlespacing*{\subsubsection}{0pt}{0.5\baselineskip}{0.5\baselineskip} 
   \usepackage{indentfirst}
   \usepackage{mathrsfs}

\newtheoremstyle{my-style}          
{\topsep}                         
{\topsep}                         
{\normalfont}                     
{\parindent}                      
{\itshape}                        
{:}                               
{.5em}                            
{}                                

\theoremstyle{my-style}
\newtheorem{proposition}{Proposition}

\makeatletter
\renewenvironment{proof}[1][\proofname]{%
	\par\pushQED{\qed} 
	\normalfont
	\partopsep0pt \itemindent0pt
	\begingroup
	\itshape #1\@addpunct{: }
	\normalfont
	\ignorespaces
}{%
	\endgroup
	\popQED
}
\makeatother

	\makeatother
	\usepackage[colorlinks,
	linkcolor=blue,
	anchorcolor=blue,
	citecolor=blue]{hyperref}
	\usepackage{xcolor}  
	\let\oldcite=\cite
	\renewcommand{\cite}[1]{\textcolor{blue}{\oldcite{#1}}}
	\let\oldeqref\eqref   
	\renewcommand{\eqref}[1]{\textcolor{blue}{\oldeqref{#1}}}   
	\let\oldref\ref
	\renewcommand{\ref}[1]{\textcolor{blue}{\oldref{#1}}}
	\newcommand{\mycitenum}[1]{\textcolor{blue}{\citenum{#1}}}  
	\usepackage{amsthm}   

	\usepackage{amsmath}
	\usepackage{amsthm}
	\theoremstyle{definition}
	\newtheorem{example}{Example}

\begin{document}
	\bstctlcite{BSTcontrol} 

	
\title{A Secure Affine Frequency Division Multiplexing System for Next-Generation Wireless Communications}
	
	\author{Ping Wang,~\IEEEmembership{Graduate Student Member,~IEEE}, Zulin Wang,~\IEEEmembership{Member,~IEEE}, \\ Yuanhan Ni,~\IEEEmembership{Member,~IEEE}, Qu Luo,~\IEEEmembership{Member,~IEEE}, Yuanfang Ma, \\Xiaosi Tian,~\IEEEmembership{Graduate Student Member,~IEEE}, and Pei Xiao,~\IEEEmembership{Senior Member,~IEEE}  
	
	\thanks{This work was supported in part by the China Postdoctoral Science Foundation under Grant Number 2024M764088; and in part by the National Natural Science Foundation of China under Grant 61971025, 62331002. An earlier version of this paper was accepted in part at IEEE ICC 2025 \cite{wang2025mine}. \textit{(Corresponding author: Yuanhan Ni.)} }
	
	\thanks{Ping Wang, Zulin Wang, Yuanhan Ni, Yuanfang Ma, and Xiaosi Tian are with the School of Electronic and Information Engineering, Beihang University, Beijing 100191, China (e-mail: wangping\_119@buaa.edu.cn; wzulin@buaa.edu.cn; yuanhanni@buaa.edu.cn; yuanfangma@buaa.edu.cn; xiaosi\_tian@buaa.edu.cn). }

	\thanks{
	 Qu Luo and  Pei Xiao  are  with the  5G \& 6G  Innovation Centre, University of Surrey, U. K. (email: q.u.luo@surrey.ac.uk; p.xiao@surrey.ac.uk). }    } 
	
	

\markboth{To be submitted to the IEEE Transactions on Wireless Communications}
	{Shell \MakeLowercase{\textit{et al.}}: A Sample Article Using IEEEtran.cls for IEEE Journals}
	
	
	\maketitle
	
	\begin{abstract}
		
	
	
	
	Affine frequency division multiplexing (AFDM) has garnered significant attention due to its superior performance in high-mobility scenarios, coupled with multiple waveform parameters that provide greater degrees of freedom for system design. This paper introduces a novel secure affine frequency division multiplexing (SE-AFDM) system, which advances prior designs by dynamically varying an AFDM pre-chirp parameter to enhance physical-layer security. 
	In the SE-AFDM system, the pre-chirp parameter is dynamically generated from a codebook controlled by a long-period pseudo-noise (LPPN) sequence.
	Instead of applying spreading in the data domain, our parameter-domain spreading approach provides additional security while maintaining reliability and high spectrum efficiency.
	We also propose a synchronization framework to solve the problem of reliably and rapidly synchronizing the time-varying parameter in fast time-varying channels.
	The theoretical derivations prove that unsynchronized eavesdroppers cannot eliminate the nonlinear impact of the time-varying parameter and further provide useful guidance for codebook design.
	Simulation results demonstrate the security advantages of the proposed SE-AFDM system in high-mobility scenarios, while our hardware prototype validates the effectiveness of the proposed synchronization framework.

	\end{abstract}

	\begin{IEEEkeywords}
	Physical layer security (PLS), Affine frequency division multiplexing (AFDM), long-period Pseudo-noise (LPPN) sequences, Synchronization.
	\end{IEEEkeywords}
	
	\section{Introduction} 
	
	
	
	
	
	

	\IEEEPARstart{S}{ixth} generation (6G) wireless communication networks are expected to support high mobility and wide area coverage scenarios, offering not only low latency of $0.1$ ms, high reliability of up to $99.99999\%$, high spectrum efficiency tripled compared to 5G, and so on, but also enhanced security \cite{ITU2023}.
	With ubiquitous connectivity, extended coverage and increased terminals expose more sensitive information to eavesdropping risks in open wireless communication environments. Furthermore, low latency and high spectrum efficiency in hyper reliable and low-latency communications
	(HRLLC) constrain the application of security strategies with high latency and high complexity \cite{ara2024physical}. Therefore, the security of 6G has attracted extensive research.
	
	Encryption implemented at the network or application layer is a well-known strategy for wireless communication security, which has been widely used in the military, medicine, and other fields. However, key distribution remains a critical challenge, particularly in decentralized and heterogeneous networks \cite{solaija2022towards}. Moreover, the high computational complexity of encryption and decryption may cause extra latency and limited throughput \cite{chen2019physical}. Thus, it is an open question for implementing encryption to meet the requirements of low latency and peak throughput in 6G, especially in networks with limited computational capabilities and constrained terminal sizes \cite{liu2022ensuring}.  
	
	
	Besides the network layer and application layer, the physical layer (PHY) can also provide wireless communication security, i.e., physical layer security (PLS) \cite{wyner1975wire}. Owing to the lower complexity compared with encryption, PLS techniques have attracted considerable research attention, including PHY key generation \cite{gao2020lightweight}, artificial noise (AN) \cite{gu2019secrecy},  secure waveform design \cite{xu2021waveform}, etc. For example, the study in \cite{gao2020lightweight} demonstrated that PHY keys can be generated from legitimate channel state information (CSI), thereby avoiding key distribution through the use of channel reciprocity.
	However, when a strong correlation exists between the legitimate and the eavesdropping channels, security performance may be degraded \cite{edman2011passive}. AN is another CSI-based method to enhance communication security. In \cite{gu2019secrecy}, based on the CSI of both the legitimate receiver and eavesdroppers, additional AN non-orthogonal to the legitimate channel was used to further improve the secrecy capacity of the communication system. Nonetheless, AN comes at the expense of precious transmission power \cite{zou2016survey}.
	
	

	
	
	

	
	
	Among various PLS techniques, secure waveform design has drawn increasing attention since it can improve communication security without incurring extra overhead.
	A representative example is the Global Positioning System (GPS), where direct sequence spread spectrum (DSSS) is applied in the data domain with a year-scale long-period pseudo-noise (LPPN) sequence, thereby providing inherent communication security.
    The essence of GPS security is that eavesdroppers are unable to synchronize with the LPPN sequence used by the transmitter \cite{becker2010secure}, while the legitimate receiver can achieve the synchronization.
	However, data-domain spread spectrum may result in a substantial reduction of spectrum efficiency.
	To address this limitation, a secure waveform based on orthogonal frequency division multiplexing (OFDM) has been proposed, benefiting from its inherent high spectral efficiency.
    In \cite{xu2021waveform}, it was revealed that wireless communication security can be enhanced by modifying the subcarrier spacing via an improved spectrum efficient frequency division multiplexing technique, thereby ensuring that eavesdroppers fail to recover the signal while legitimate receivers succeed.
    However, the bit error rate (BER) performance of the OFDM-based waveform deteriorates due to the intercarrier  interference (ICI) induced by fast-varying channels under high-mobility scenarios \cite{yuan2024papr}.
    

    


	
	
	
	To achieve both reliability and security in high-mobility scenarios, various orthogonal time frequency space (OTFS)-based secure waveforms have been investigated \cite{sun2021orthogonal,sun2021secure}, taking advantage of the ability of OTFS to achieve full diversity in fast time-varying channels \cite{OTFShadani2017orthogonal}.
	Specifically, by spreading the information symbols in either the delay or Doppler domain, a secure OTFS-based waveform, namely DS-OTFS, was developed to achieve PLS at the cost of reduced spectrum efficiency \cite{sun2021orthogonal}.
	Moreover, a rotated orthogonal time frequency space (R-OTFS) waveform was proposed to enhance security by rotating information symbols based on the legitimate channel, thereby reducing the signal-to-interference-plus-noise ratio (SINR) for the eavesdropper \cite{sun2021secure}.
	However, for R-OTFS, once a strong correlation exists between the legitimate and the eavesdropping channels, the security may be compromised.
	Furthermore, due to the excessive pilot overhead caused by the two-dimensional (2D) structure of OTFS, OTFS-based secure waveforms improve communication security with a loss of spectrum efficiency.
	

	
	
	
	
	Recently, affine frequency division multiplexing (AFDM) with one-dimensional (1D) pilots has emerged as a promising solution for high-mobility communications, owing to its reliability, efficiency, and greater design flexibility \cite{bemani2023affine}.
	By adjusting two pre-chirp parameters, i.e., $c_1$ and $c_2$, AFDM can be fully compatible with OFDM, which makes AFDM regarded as an attractive candidate for 6G \cite{zhou2024overview}. AFDM-based research has investigated improving the reliability and spectrum efficiency of communications by tuning $c_1$, such as channel estimation \cite{yin2022pilot}, equalization \cite{bemani2022low}, sparse code multiple access \cite{luo2024afdm},  multiple-input multiple-output system \cite{yin2024diagonally}, integrated sensing and communications \cite{ni2025integrated}, etc. Meanwhile, several studies have investigated adjusting parameter $c_2$ to reduce PAPR \cite{yuan2024papr} and to enable index modulation \cite{tao2025affine,zhu2023design,liu2025pre}.






    Since adjusting the pre-chirp parameter $c_2$ to enhance security is an inherent advantage of AFDM,
	in parallel with our earlier conference version \cite{wang2025mine}, a few studies have begun to focus on enhancing communication security by tuning the AFDM parameter $c_2$ \cite{rou2025chirp,rou2025chirp2,tek2025novel}.
	In \cite{rou2025chirp} and \cite{rou2025chirp2}, it was shown that communication security can be enhanced by sharing a fixed permutation of $c_2$ values across subcarriers between the transmitter and the legitimate receiver, under the assumption that the eavesdropper is unaware of the permutation. However, the static parameter set poses a potential risk of inference by eavesdroppers.
    In \cite{tek2025novel}, channel reciprocity was exploited to enhance communication security by generating dynamic $c_2$ values from the time-domain channel matrix. In contrast to \cite{rou2025chirp}, $c_2$ in \cite{tek2025novel} is dynamically varied across different subcarriers of different AFDM symbols.  Nevertheless, the security performance in \cite{tek2025novel} may be degraded when the legitimate and eavesdropping channels are highly correlated \cite{edman2011passive}.  Meanwhile, 
    the channel reciprocity may not hold under the frequency-division duplex (FDD) scheme,  thereby making CSI feedback particularly challenging in fast time-varying channels.
    Although initial attempts have been made to modify $c_2$  for enhancing security, the dynamic adjustment mechanism and the corresponding synchronization method of $c_2$ in fast time-varying channels remain largely open problems.

	To this end, this paper presents a secure affine frequency division multiplexing (SE-AFDM) system that varies and synchronizes the parameter $c_2$ using an LPPN sequence to ensure reliability, security, and high spectrum efficiency.
	The time-varying $c_2$ is generated from a codebook controlled by the LPPN sequence, 
	which allows legitimate receivers to synchronize $c_2$ with low complexity, while preventing eavesdroppers from achieving synchronization. 
	The impacts of dynamic $c_2$ on both the legitimate receivers and the eavesdroppers are analyzed, providing valuable insights for codebook design. To synchronize dynamic $c_2$ at the legitimate receivers, a synchronization framework is proposed, which includes a frame structure and a synchronization strategy. 
	Moreover, the proposed framework is verified through over-the-air propagation experiments using a software-defined radio (SDR) platform.
	Simulation and experimental results confirm the security of the SE-AFDM system and validate the effectiveness of the proposed synchronization framework.
	For clarity, we summarize our contributions as follows:
	\begin{itemize}
		\item We propose an SE-AFDM wireless communication system by introducing an LPPN sequence to dynamically generate the parameter $c_2$. In the SE-AFDM system, only a few fixed configuration parameters of the LPPN sequence generator remain confidential, whereas the codebook of $c_2$ and the state parameters of the LPPN sequence generator are public for both legitimate receivers and eavesdroppers.
        Unlike applying DSSS in the data domain, the control of parameter $c_2$ via an LPPN sequence can be regarded as a parameter-domain spreading, which ensures high reliability and security while maintaining spectral efficiency.

        \item We derive the impact of time-varying $c_2$ on the security performance.
        We prove that the time-varying $c_2$ can be eliminated at the legitimate receiver with synchronized $c_2$. However, eavesdroppers cannot separate the time-varying $c_2$ from the random information symbols without synchronization of the dynamic $c_2$, thereby enhancing communication security. Moreover, we reveal that the SINR at the eavesdropper is a function of the variation range and the number of candidate values of $c_2$, thereby providing useful guidelines for designing the codebook.

        \item 
         We design a synchronization framework with low complexity for dynamic $c_2$ at legitimate receivers. In such a framework, a frame structure is designed to facilitate $c_2$ synchronization while simultaneously ensuring the secure transmission of valuable information. Based on this, a corresponding synchronization strategy is proposed, where the legitimate receiver achieves $c_2$ synchronization by synchronizing the LPPN sequence. Experimental results verify the effectiveness of our proposed synchronization framework.

	\end{itemize}
	
	

	The rest of this paper is organized as follows. Section \ref{sec:section2} briefly introduces the AFDM communication model and the generation principle of the LPPN sequence. In Section \ref{sec:section3}, we present an SE-AFDM system. The security of the SE-AFDM system is theoretically analyzed in Section \ref{sec:security}. Section \ref{sec:syn_strategy} proposes a synchronization framework for the SE-AFDM system. 
	Simulation and experimental results verify the security of the proposed SE-AFDM system and the effectiveness of the proposed synchronization framework in Section \ref{sec:results}.
	Finally, Section \ref{sec:conclusion} concludes this paper.

	
	


	\textit{Notation}: Throughout the paper, $\mathbf{X}$, $\mathbf{x}$, and $x$ denote a matrix, vector, and scalar, respectively. $\mathbf{I}_{N}$ denotes an $N \times N$ identity matrix.
	$\left\lfloor  \cdot  \right\rfloor $, $\langle\cdot\rangle_{N}$, $\vert\cdot\vert$, $\odot$, $\oplus$,  $(\cdot)^{*}$, $(\cdot)^{T}$ and $(\cdot)^{{H}}$ are the floor function, the modulo $N$ operator, cardinality operator, the Hadamard product, the modulo-2 addition operator, the conjugate operation, the transpose, and the Hermitian transpose, respectively. ${\mathbb{E}(\cdot)}$ is the expectation operator.
	${\rm diag} (\mathbf{x})$ forms a diagonal matrix with the elements of $\mathbf{x}$ on the main diagonal. The operator $\Re$ is used to extract the real part of a complex number.   ${\rm gcd}(a,b)$ is the greatest common divisor of $a$ and $b$. $\mathcal{M}_{R \text {-QAM }}(\cdot) \text { denotes the } R \text {-QAM mapping function}$. $\operatorname{vec}(\cdot)$ denotes the vectorization operator.    
	
	


	
	

	\section{Preliminaries} \label{sec:section2}
	
	
	\subsection{AFDM Communication Model}
	
	Firstly, the AFDM model proposed in \cite{bemani2023affine} is briefly reviewed. Let $\mathbf{x}$ denote an $N {\times} 1$ vector of quadrature amplitude modulation (QAM) symbols. By performing the $N$-point inverse discrete affine Fourier transform (IDAFT), $\mathbf{x}$ is mapped from the discrete affine Fourier transform (DAFT) domain to the time domain, i.e., \cite{bemani2023affine}
	\begin{equation} \label{sn}
		{s}\left[{n}\right] = \frac{1}{{\sqrt N }}\sum\limits_{m = 0}^{N - 1} {{x}\left[ {m} \right]} {e^{j2\pi \left( {{c_1}{n^2} + \frac{mn}{N} + {c_2}{m^2}} \right)}} ,
	\end{equation}
	where $c_1$ and $c_2$ are two DAFT chirp parameters, and $n=0,\ldots,N {-} 1$. After adding a chirp-periodic prefix (CPP) of length $N_{\rm cp}$, \eqref{sn} is rewritten as \cite{bemani2023affine}	
	\begin{equation}\label{eq:symbol_AFDM}
		{s}\left[ {n} \right] {=} {s}\left[ {n {+} N} \right]{e^{ - i2\pi {c_1}\left( {{N^2} + 2Nn} \right)}},n =  - {N_{\rm cp}}, \ldots , - 1.
	\end{equation}
	
	
	
	Then, the AFDM signal is transmitted over a communication channel with $P$ paths, in which the gain coefficient, time delay and Doppler shift of the $i$-th path are denoted by ${h_i}$, ${\tau _i}$, ${f_{d,i}}$, respectively. The received signal in the time domain is given by 
	[\mycitenum{wu2022integrating}, Eq.~(6)]
	\begin{equation}\label{eq:received_sig_time}
		{r}\left[ n \right] = \sum\limits_{i = 1}^P {{{\tilde h}_i}} {s}\left[ {n - {l_i}} \right]{e^{j2\pi {f_i}n}} + {{w}}_t\left[ n \right],
	\end{equation}
	where $\mathbf{{w}}_t \hspace{-0.5ex}\sim \hspace{-0.5ex}\mathcal {CN}\left( {0, \sigma_{c} ^2\mathbf{I}} \right)$ is an additive white Gaussian noise (AWGN) vector, ${{\tilde h}_i} {=} {h_i}{e^{ - j2\pi {f_{d,i}}{\tau _i}}}$, ${l_i} {=} {{{\tau _i}} \mathord{\left/
			{\vphantom {{{\tau _i}} {{t_s}}}} \right.
			\kern-\nulldelimiterspace} {{t_{\rm s}}}}$, ${f_i} {=} {f_{d,i}}{t_{\rm s}}$ with ${t_{\rm s}}$ being the sampling interval, and $n \hspace{-0.5ex}\in \hspace{-0.5ex}\left[ { - {N_{\rm cp}},N - 1 } \right]$.

	After discarding CPP and performing $N$-point DAFT, the resulting signal in the DAFT domain can be written as \cite{bemani2023affine}
	\begingroup
	\setlength{\abovedisplayskip}{3pt}
	\setlength{\belowdisplayskip}{3pt}
	\begin{equation}
		\mathbf{y} = {\mathbf{H}_{\rm eff}}\mathbf{x} + \mathbf{{w}}_a= \sum\limits_{i = 1}^P {{\tilde h_i}{\mathbf{H}_{i,\rm A}}\mathbf{x}} + \mathbf{{ w}}_a,
	\end{equation}
	\endgroup
	where ${\mathbf{H}_{\rm eff}} {=} \mathbf{A}{\mathbf{H}_{c,t}}\mathbf{A}^{{H}}$ denotes the effective channel matrix and ${\mathbf{H}_{c,t}}$ denotes the time-domain channel matrix, i.e., ${\mathbf{H}_{c,t}} = \sum\limits_{i = 1}^P {{\tilde h_i}{{\bf{\Gamma }}_{{\rm cpp}{_i}}}{\bm{\Delta} _{{f_i}}}{\bm{\Pi} ^{\left({l_i}\right)}}}$, $\bm{\Pi}$ denotes the forward cyclic-shift matrix, ${\bm{\Delta} _{{f_i}}} = {\rm diag}\left( {{e^{ i2\pi {f_i}n}},  n = 0, \ldots ,  N {-} 1} \right)$,  ${\mathbf{H}_{i, \rm A}}=\mathbf{A}{{\bf{\Gamma }}_{{\rm cpp}{_i}}}{\bm{\Delta} _{{f_i}}}{\bm{\Pi} ^{\left({l_i}\right)}}\mathbf{A}^{{H}}$, $\mathbf{A} = {\bm{\Lambda} _{{c_2}}}\mathbf{F}{\bm{\Lambda} _{{c_1}}}$, $\mathbf{{w}}_a=\mathbf{A}\mathbf{w}_t$,  $\mathbf{F}$ is the discrete Fourier transform (DFT) matrix,  ${\bm{\Lambda} _{{c_i}}} {=}  {\rm diag} {\left( {{e^{ - j2\pi c_i{n^2}}},n = 0 ,\ldots ,N {-} 1}, i=1,2 \right)}$, and ${{\bf{\Gamma }}_{{\rm cpp}{_i}}}$ is a diagonal matrix, which is defined as
	\begin{eqnarray}
		{{\bf{\Gamma }}_{{\rm cpp}{_i}}} = {\rm diag}\left( {\left\{ {\begin{array}{*{20}{l}}
					{{e^{ - i2\pi {c_1}\left( {{N^2} - 2N\left( {{l_i} - n} \right)} \right)}}},&{n < {l_i}},\\[-2pt]
					1,&{n \ge {l_i}}.
			\end{array}} \right.} \right),
	\end{eqnarray}
	and ${{H}_{i,\rm A}}\hspace{-0.5ex}\left[ {p,q} \right]$ is given by \cite{bemani2023affine}
	\begin{eqnarray}
		{{H}_{i, \rm A}}\left[ {p,q} \right] = \frac{1}{N}{e^{j\frac{{2\pi }}{N}\left( {N{c_1}l_i^2 - ql_i + N{c_2}\left( {{q^2} - {p^2}} \right)} \right)}}{\mathcal{F}}_i\left[ {p,q} \right],
	\end{eqnarray}
	where ${{\mathcal{F}}_i}\left[ {p,q} \right] {=} \frac{{{e^{ - j2\pi \left( {p - q - {\nu_i} + 2N{c_1}{l_i}} \right)}} - 1}}{{{e^{ - j\frac{{2\pi }}{N}\left( {p - q - {\nu_i} + 2N{c_1}{l_i}} \right)}} - 1}}$ with ${\nu _i} = N{f_i} = \frac{{{f_{d,i}}}}{{\Delta f}} = {\alpha _i} + {a_i} \in \left[ { - {\nu _{\max }}, {\nu _{\max }}} \right]$ denotes the Doppler shift normalized by the subcarrier spacing ${\Delta f}$, $p = 0,\ldots N-1$, $q = 0,\ldots N-1$, ${\alpha _i} {\in} \left[ { - {\alpha _{\max }},{\alpha _{\max }}} \right]$ and ${a_i} {\in} \left( { - \frac{1}{2},\frac{1}{2}} \right]$ correspond to the integral and fractional part of ${\nu _i}$, respectively. 
	
%

	\subsection{Generation Principle of LPPN Sequence}
	
	Following the generation mechanism of the GPS P-code \cite{IS-GPS}, we next review the LPPN sequence generation method based on the precession and modulo-2 addition of short-cycled sequences.
	The sequence nomenclature in this section follows the specifications outlined in \cite{IS-GPS}.
	As shown in Fig. \ref{fg:P_generation}, the LPPN sequence $\mathrm{L}$ is generated from sequences $\mathrm{X1}$ and $\mathrm{X2}$ via precession and modulo-2 addition. The sequence $\mathrm{X1}$ is calculated using short sequences $\mathrm{X1A}$ and $\mathrm{X1B}$, whereas $\mathrm{X2}$ is obtained from short sequences $\mathrm{X2A}$ and $\mathrm{X2B}$.
	


 	\begin{figure}[H]
 	\vspace{-5pt}
 	\centering
 	\includegraphics[width=3.3in]{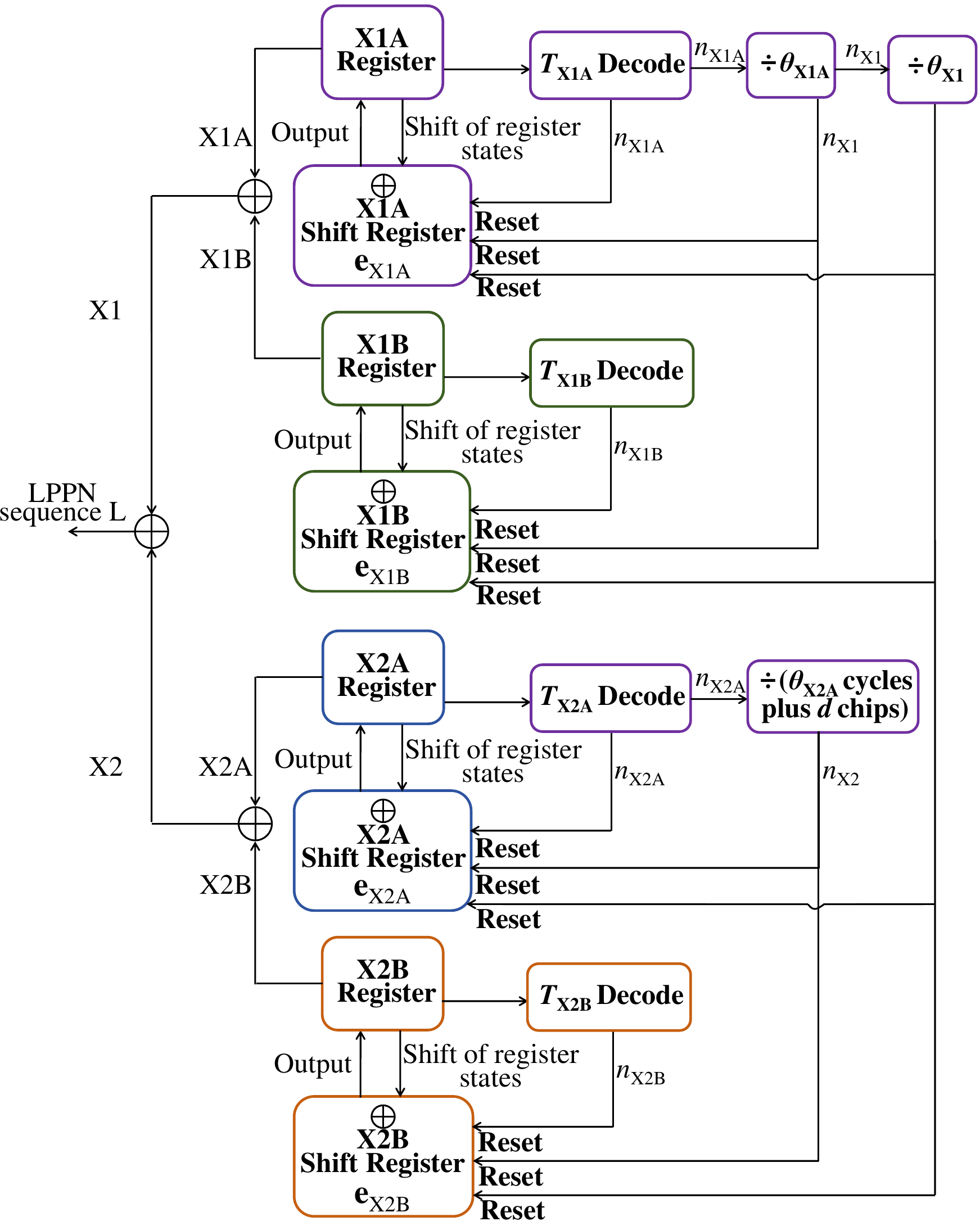}
 	\caption{The mechanism of LPPN sequence generator \cite{IS-GPS}.  
 		\label{fg:P_generation}}
 	\vspace{-5pt}
 \end{figure}

	The short sequences $\mathrm{X1A}$, $\mathrm{X1B}$, $\mathrm{X2A}$, and $\mathrm{X2B}$ are generated by four $S$-stage shift registers, respectively.
	The polynomial coefficient vector of the shift register of ${\beta}$ sequence, where $\beta \in \{ {{\mathrm{X1A}, \mathrm{X1B}, \mathrm{X2A}, \mathrm{X2B}}}\}$, is represented by ${\bf{e}}_{\beta} = [e_{{\beta},1}, \ldots, e_{\beta,S}]^T\in {\left\{ {0,1} \right\}}^{S \times 1}$.
	 Accordingly, the state vector of the shift register of ${\beta}$ sequence at the $k$-th time step is denoted as ${\bf{s}}^k_{\beta} = [s^k_{{\beta},1}, \ldots, s^k_{{\beta},S}]^T\in {\left\{ {0,1} \right\}}^{S \times 1}$.
    All four registers perform feedback computation and state shifting. Taking one cycle of the  shift register of the $\mathrm{X1A}$ sequence as an example, the output at the $k$-th time step is given by the register state at the $S$-th stage, denoted as $\mathrm{X} 1 \mathrm{A}[k] = s^k_{{\rm X1A},S}$. At time step $k+1$, the new state value of the first-stage register is computed by \cite{IS-GPS} 
   \begingroup
   \setlength{\abovedisplayskip}{3pt}
   \setlength{\belowdisplayskip}{3pt}
    \begin{eqnarray} \label{eq:X1A}
   	s^{k+1}_{{\rm X1A},1} = \left\langle \sum\limits_{i = 1}^{S} e_{{\rm X1A},i} \cdot s^{k}_{{\rm X1A},i} \right\rangle_2. 
   \end{eqnarray}
   \endgroup
   Then, the register states are shifted by $s^{k+1}_{{\rm X1A},i+1} = s^{k}_{{\rm X1A},i}$ for $i=1,\ldots,S-1$, resulting in the new state vector ${\bf{s}}^{k+1}_{{\rm X1A}} = [s^{k+1}_{{\rm X1A},1}, \ldots, s^{k+1}_{{\rm X1A},S}]^T$. The chip of the $\mathrm{X1A}$ sequence at time step $k+1$ is represented as 
   \begingroup
    \setlength{\abovedisplayskip}{5pt}
   \setlength{\belowdisplayskip}{4pt}
   \begin{eqnarray}
   \mathrm{X} 1 \mathrm{A}[k+1]=  s^{k+1}_{{\rm X1A},S}. 
   \end{eqnarray}
   \endgroup

   To achieve the precession between $\mathrm{X1A}$ and $\mathrm{X1B}$ sequences as well as between the $\mathrm{X2A}$ and $\mathrm{X2B}$ sequences, the cycles of the $\mathrm{X1A}$, $\mathrm{X1B}$, $\mathrm{X2A}$ and $\mathrm{X2B}$ sequences are typically shortened and denoted as $T_{\rm X1A}$, $T_{\rm X1B}$, $T_{\rm X2A}$ and $T_{\rm X2B}$.  
   The corresponding cycle counts are denoted by $n_\beta \!\in \!\{ {0, \ldots ,{\theta _{{\beta}}}} \}$ with $\beta \! \in\! \{ {{\mathrm{X1A}, \mathrm{X1B}, \mathrm{X2A}, \mathrm{X2B}}}\}$. 
   Typically, $T{_{{\rm{X1A}}}} = T{_{{\rm{X2A}}}} < T{_{{\rm{X1B}}}} = T{_{{\rm{X2B}}}}$, ${\rm gcd}$$(T_{\rm X1A},T_{\rm X1B}) =$ ${\rm gcd}$$(T_{\rm X2A},T_{\rm X2B}) =1$,  ${\theta _{{\rm{X1A}}}} = {\theta _{{\rm{X2A}}}} \le {T_{{\rm{X2B}}}}$, ${\theta _{{\rm{X1B}}}} = \left\lfloor {\frac{{{T_{{\rm{X1A}}}}{\theta _{{\rm{X1A}}}}}}{{{T_{{\rm{X1B}}}}}}}\right\rfloor,$ and ${\theta _{{\rm{X2B}}}} = \left\lfloor {\frac{{{T_{{\rm{X2A}}}}{\theta _{{\rm{X2A}}}}}}{{{T_{{\rm{X2B}}}}}}} \right\rfloor$ \cite{IS-GPS}.

   	\begin{figure*}[!t]  
	
	\begin{align}
		{\rm X1}[k] &=  
		\begin{cases}
			{\rm X1A}[k - T_{\rm X1}n_{\rm X1} - T_{\rm X1A}n_{\rm X1A}] \oplus 
			{\rm X1B}[k - T_{\rm X1}n_{\rm X1} - T_{\rm X1B}n_{\rm X1B}], & {\rm{C}}.1, \\
			{\rm X1A}[k - T_{\rm X1}n_{\rm X1} - T_{\rm X1A}n_{\rm X1A}] \oplus 
			{\rm X1B}[T_{\rm X1B}], & {\rm{C}}.2.
		\end{cases} \label{X1} \\  		 
		{\rm X2}[k] &= 
		\begin{cases}
			{\rm X2A}[k - T_{\rm X2}n_{\rm X2} - T_{\rm X2A}n_{\rm X2A}] \oplus
			{\rm X2B}[k - T_{\rm X2}n_{\rm X2} - T_{\rm X2B}n_{\rm X2B}], & {\rm{C}}.3, \\
			{\rm X2A}[k - T_{\rm X2}n_{\rm X2} - T_{\rm X2A}n_{\rm X2A}] \oplus
			{\rm X2B}[T_{\rm X2B}], & {\rm{C}}.4, \\
			{\rm X2A}[T_{\rm X2A}] \oplus {\rm X2B}[T_{\rm X2B}], & {\rm{C}}.5.
		\end{cases} \label{X2}
	\end{align}
	\hrule
\end{figure*}

   Through the precession and modulo-2 addition between the $\mathrm{X1A}$ and $\mathrm{X1B}$ sequences, the $\mathrm{X1}$ sequence is generated by \eqref{X1}, which is shown at the top of this page. In \eqref{X1}, ${\rm{C}}.1$ and ${\rm{C}}.2$ are respectively given by  
    \begingroup
   \setlength{\abovedisplayskip}{7pt}
   \setlength{\belowdisplayskip}{7pt}
   \begin{subequations} \label{cons_13}
   	\begin{align}
   		&{\rm{C}}.1: k\! \le\! {T_{{\rm{X1}}}}{n_{{\rm{X1}}}} \!\! + \!\!{T_{{\rm{X1B}}}}	{\theta _{{\rm{X1B}}}},\label{cons:1}\\
   		&{\rm{C}}.2:{T_{{\rm{X1}}}}{n_{{\rm{X1}}}} \!\!+ \!\!{T_{{\rm{X1B}}}}{\theta _{{\rm{X1B}}}} \! <\! k \!\le\! {T_{{\rm{X1}}}}{n_{{\rm{X1}}}}\!\! +\!\! {T_{{\rm{X1}}}},\label{cons:2}
   	\end{align}
   \end{subequations} 
    \endgroup
   where $k \in \left\{ {0, \ldots ,{T_{\rm{L}}-1}} \right\}$ is the chip index, $n_{\rm X1} \in \left\{ {0,\ldots,{\theta _{{\rm{X1}}}}} \right\}$ is the count of the $\mathrm{X1}$ sequence cycles, $T_{\rm X1} = T_{\rm X1A}\theta_{\rm X1A}$ denotes the cycle of the $\mathrm{X1}$ sequence, and ${T_{\rm{L}}} = {T_{{\rm{X1}}}}{\theta _{{\rm{X1}}}}$ represents the cycle of the LPPN sequence. After each generation of $\mathrm{X1}$ sequence, the shift registers of the $\mathrm{X1A}$ and $\mathrm{X1B}$ sequences are reinitialized. When $n_{\rm X1}$ reaches ${\theta _{{\rm{X1}}}}$, a full cycle of the LPPN sequence is completed and all counters and registers are reset. 
   Notably, as illustrated in \eqref{X1} and \eqref{cons:2}, once the shift register of the $\mathrm{X1B}$ sequence  completes its ${\theta _{{\rm{X1B}}}}$-th cycle within each $\mathrm{X1}$ sequence period, it remains in its final state until the start of the next $\mathrm{X1}$ sequence cycle. To avoid periodic repetition in the $\mathrm{X1}$ sequence, it is generally required that $\theta{_{{\rm{X1A}}}} - \frac{{T{_{{\rm{X1}}}}}}{{T{_{{\rm{X1B}}}}}}  \le  T{_{{\rm{X1B}}}} - T{_{{\rm{X1A}}}}$ \cite{IS-GPS}.
    

   
    
    

     Similarly, as given in \eqref{X2}, $\mathrm{X2}$ sequence is generated by the modulo-2 addition of the $\mathrm{X2A}$ and $\mathrm{X2B}$ sequences, where  
      \begingroup
     \setlength{\abovedisplayskip}{7pt}
     \setlength{\belowdisplayskip}{7pt}
     \begin{subequations} \label{Cons:X2}
    \begin{align}
     &{\rm{C}}.3:k \!\! \le\!\! {T_{{\rm{X2}}}}{n_{{\rm{X2}}}}\! +\! {T_{{\rm{X2B}}}}{\theta _{{\rm{X2B}}}}, \label{X2a}\\
     &{\rm{C}}.4:{T_{{\rm{X2}}}}{n_{{\rm{X2}}}}\! +\! {T_{{\rm{X2B}}}}{\theta _{{\rm{X2B}}}} \!\!<\! \!k \le {T_{{\rm{X2}}}}{n_{{\rm{X2}}}}\! + \!{T_{{\rm{X2A}}}}{\theta _{{\rm{X2A}}}},\label{X2b}\\
     &{\rm{C}}.5:{T_{{\rm{X2}}}}{n_{{\rm{X2}}}}\! +\! {T_{{\rm{X2A}}}}{\theta _{{\rm{X2A}}}}\! \!< \!\! k \le {T_{{\rm{X2}}}}{n_{{\rm{X2}}}} \!+\! {T_{{\rm{X2}}}}\label{X2c},
     \end{align}
    \end{subequations}
    \endgroup
   where $T_{\rm X2} =  {{T_{{\rm{X2A}}}}{\theta _{{\rm{X2A}}}} + d}$ denotes the cycle of the $\mathrm{X2}$ sequence, ${n_{{\rm{X2}}}} \in \left\{ {0, \ldots ,{\theta _{{\rm{X2}}}}}\right\}$ is the count of the $\mathrm{X2}$ sequence cycles, and ${\theta _{{\rm{X2}}}} = \left\lfloor {\frac{{{T_{{\rm{X1}}}}{\theta _{{\rm{X1}}}}}}{{{T_{{\rm{X2}}}}}}}\right\rfloor.$ 
    When $k$ satisfies \eqref{X2a}, $\mathrm{X2}$ sequence is generated by the modulo-2 addition of the $\mathrm{X2A}$ and $\mathrm{X2B}$ sequences during precession.
    Once $k$ meets \eqref{X2b}, the shift register of the $\mathrm{X2B}$ sequence holds its final state. In addition, to introduce the precession between the $\mathrm{X1}$ and $\mathrm{X2}$ sequences, the reset of the shift registers of the $\mathrm{X2A}$ and $\mathrm{X2B}$ sequences is delayed by $d$ chips, as shown in \eqref{X2c}.   Typically, ${\rm gcd} \left( {{T_{{\rm{X1}}}},{T_{{\rm{X2}}}}} \right) = 1$ \cite{IS-GPS}. To ensure a long cycle of the final sequence, ${\theta _{{\rm{X1}}}} = {T_{{\rm{X2}}}}$ and ${\theta _{{\rm{X2}}}} = {T_{{\rm{X1}}}}$.

	By performing modulo-2 addition between the $\mathrm{X1}$ and $\mathrm{X2}$ sequences, the $k$-th chip of the LPPN sequence is given by
	\begin{equation} \label{P}
		\mathrm{L}[k]=\mathrm{X} 1\left[k\right] \oplus \mathrm{X} 2\left[k\right], k = 0, \ldots ,{T_{{\rm{L}}}-1}.
\end{equation}


The parameters of the LPPN sequence generator can be categorized into configuration parameters and state parameters. The configuration parameters include ${\bf{e}}_{\rm X1A}$, ${\bf{e}}_{\rm X1B}$, ${\bf{e}}_{\rm X2A}$, ${\bf{e}}_{\rm X2B}$, $T_{\rm X1A}$, $T_{\rm X1B}$, $T_{\rm X2A}$, $T_{\rm X2B}$, ${\theta _{{\rm{X1A}}}}$, ${\theta _{{\rm{X2A}}}}$, ${\theta _{{\rm{X1}}}}$, and $d$, while the state parameters include the counter values ${n_{{\rm{X1A}}}}$, ${n_{{\rm{X1B}}}}$, ${n_{{\rm{X2A}}}}$, ${n_{{\rm{X2B}}}}$, ${n_{{\rm{X1}}}}$, ${n_{{\rm{X2}}}}$ and the current register states ${\bf{s}}^k_{\rm{X1A}}$, ${\bf{s}}^k_{\rm{X1B}}$, ${\bf{s}}^k_{\rm{X2A}}$, ${\bf{s}}^k_{\rm{X2B}}$. From \eqref{eq:X1A}-\eqref{P}, it can be observed that  eavesdroppers cannot reconstruct the LPPN sequence without knowledge of the configuration parameters, even if the state parameters are revealed, which motivates us to propose the SE-AFDM system based on the LPPN sequence.

	\section{Secure Affine Frequency Division Multiplexing} \label{sec:section3}
	
In this section, we introduce an SE-AFDM system in which $c_2$ is dynamically tuned based on codebook indices driven by the LPPN sequence. Then, we derive the input-output relationship of SE-AFDM at the legitimate receiver. Based on theoretical analysis, we show that the impact of the time-varying parameter $c_2$ can be compensated using the derived effective channel matrix for the legitimate user without performance degradation compared with the existing AFDM system.
	
	
	
	
	\subsection{Proposed SE-AFDM System}
	
	In this paper, we use Alice, Bob and Eve to denote the transmitter, the legitimate receiver, and the eavesdropper, respectively.


	
	
	\subsubsection{ Modulation with dynamic $c_2$ at Alice }
	
	\begin{figure*}[htbp]
		\centering	
		\includegraphics[width=7.0in]{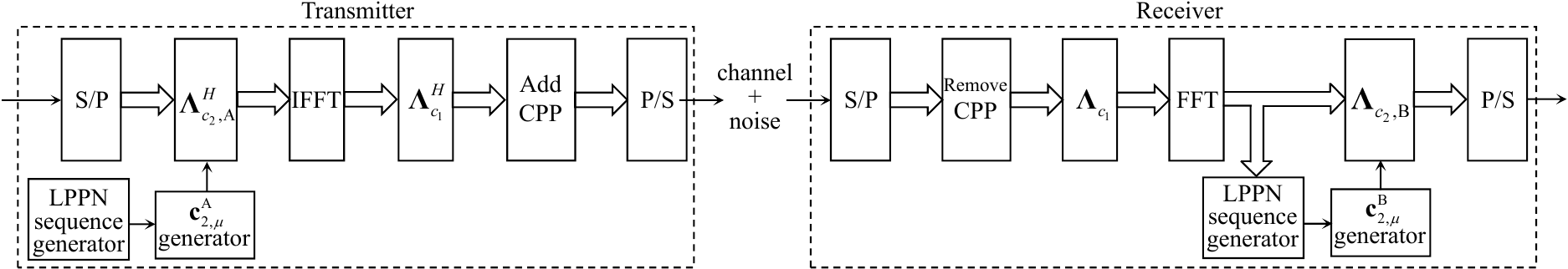}
		\caption{Block diagram of SE-AFDM communication system. 
			\label{fg:SE_AFDM}}
		\vspace{-10pt}
	\end{figure*}
	
	
	
	The corresponding block diagram of the proposed SE-AFDM system is shown in Fig. \ref{fg:SE_AFDM}.
	Firstly, the information symbol vector $\mathbf{x}$ is multiplied by a dynamic security-aware matrix ${\bm{\Lambda} ^{H}_{{c_2,\rm A}}} {=} {\rm diag}\left( {{e^{ j2\pi {c}^{\rm A}_{2,\mu}\left[m\right]{m^2}}},m = 0, \ldots ,N {-} 1} \right)$, where $\mathbf{c}^ {\rm A}_{2,\mu}$ is an ${N \times 1}$ parameter vector at Alice, and ${c}^{\rm A}_{2,\mu}\left[m\right]$ denotes the parameter $c_2$ corresponding to the $m$-th subcarrier of the $\mu$-th AFDM symbol. As shown in Fig. \ref{fg:LPPN_index}, ${c}^{\rm A}_{2,\mu}\left[m\right]$ is generated from a codebook $\mathcal{C}_2 $ under the control of the LPPN sequence. 
	The codebook $\mathcal{C}_2$ is pre-designed by uniformly discretizing $c_2$ over [$-c_{2,\rm max}$, $c_{2,\rm max}$], where $c_{2, \max }$ is the maximum value of $c_2$. Then, we have the codebook 
	 \begingroup
	\setlength{\abovedisplayskip}{3pt}
	\setlength{\belowdisplayskip}{3pt}
	\begin{equation}
		\mathcal{C}_2 = \left\{ {{A_0},{A_1}, \cdots ,{A_{M-1}}} \right\},
	\end{equation}
	\endgroup
	where $M$ denotes the number of $c_2$ candidates, i.e., $M= \vert \mathcal C_2 \vert$. In this paper, the $k$-th element ${A_k}$ of the codebook is obtained by  
    \begingroup
	\setlength{\abovedisplayskip}{3pt}
	\setlength{\belowdisplayskip}{3pt}
	\begin{equation} \label{eq:Az}
		A_{k}={\left\{\begin{array}{ll}
				-c_{2, \max }, &M=1, \\[-2pt]
				-c_{2, \max }+k {\Delta_{\rm A}}, &M \geq 2,
			\end{array}\right.}
	\end{equation} 
	\endgroup
	where $k = 0, \ldots,M-1$, and ${\Delta_{\rm A}}=\frac{{2{c_{2,\max }}}}{M-1}$ denotes the codebook interval. By sequentially truncating the LPPN sequence $\rm L$ every ${\log _2}M$ elements, the $\varphi$-th truncated LPPN sequence can be written as
	\begin{equation} \label{eq:L_M} 
		{\mathrm{L}_{\varphi} =\left\{{\rm L}[{\varphi-\log_{2} M+1}],\ldots, {\rm L}[\varphi]\right\},}
	\end{equation}
	where $\varphi \! =\! \mu N \!+ \! m \! \in \! \left\{ {0, \ldots ,{T_{\rm{L}}-1}} \right\}$, ${{\rm {L}}[i]} \!= \!1$ for $i\! \in \!\{  \!- ({\log _2}M) \!+\! 1, \ldots ,-1\}$, and ${\log _2}M$ is the length of the truncated LPPN sequence. To generate the index of ${c}^{\rm A}_{2,\mu}\left[m\right]$, we convert the binary sequence $\mathrm{L}_{\varphi}$ to a decimal number, which is given by 
	 \begingroup
	\setlength{\abovedisplayskip}{3.5pt}
	\setlength{\belowdisplayskip}{3.5pt}  
	\begin{equation} \label{z}
		k = \sum\limits_{z = 0}^{{{\log }_2}M-1} {{\rm{L}}_\varphi[z]}  {2^{{{\log }_2}M -1- z}}.  
	\end{equation}  
	\endgroup 
With the index $k$ obtained from \eqref{z}, ${c}^{\rm A}_{2,\mu}\left[m\right]$ is selected as
	 \begingroup
\setlength{\abovedisplayskip}{3.5pt}
\setlength{\belowdisplayskip}{3.5pt}
	\begin{equation}
		{{c}}_{2,\mu}^{\rm A}[m] = {A_k},
	\end{equation}   
		\endgroup  

	 As the LPPN sequence is generated continuously, different $c_2$ values are dynamically produced. The codebook $\mathcal{C}_2$ is available to Alice, Bob, and even Eve. The mapping between the LPPN sequence and the dynamic $c_2$ indicates that synchronizing $c_2$ between Alice and Bob essentially reduces to synchronizing their LPPN sequences.



	
    \begin{figure}[htbp]
		\vspace{-5pt}
		\centering
		\includegraphics[width=3.5in]{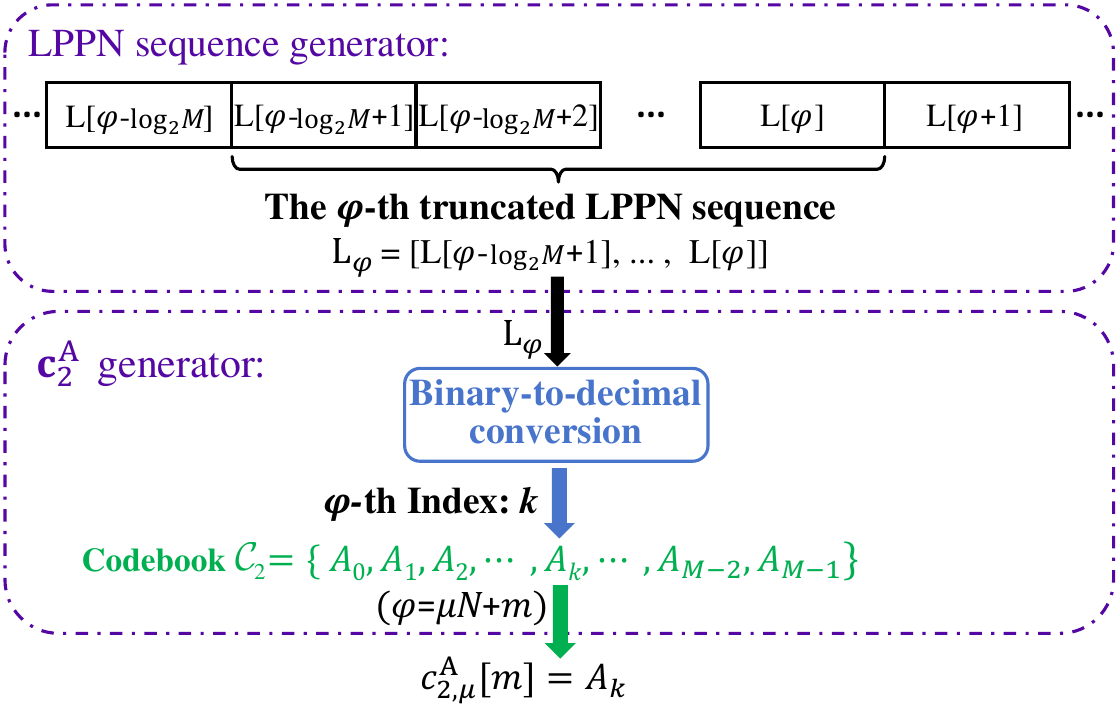}
		\caption{The generation procedure of dynamic $c_2$ based on the LPPN sequence and the codebook.
			\label{fg:LPPN_index}}
		\vspace{-5pt}
	\end{figure}

	
	
	Then, the subsequent operations are the same as the existing AFDM modulation process, i.e., performing IDFT, multiplying by the matrix ${\bm{\Lambda} ^{H}_{{c_1}}}$, and adding a CPP. The resulting SE-AFDM waveform at Alice in the time domain can be written as
	\vspace{-0.5ex}
	\begin{equation} \label{eq:ES_AFDM_Tx}
		{s}_{\rm A}\left[ n \right]=\frac{1}{\sqrt{N}}\sum\limits_{m=0}^{N-1}{{x}}\left[ m \right]{{e}^{j2\pi \left( {{c}_{1}}{{n}^{2}}+{c}^{\rm A}_{2,\mu}\left[m\right]{{m}^{2}}+\frac{mn}{N} \right)}} ,
	\end{equation}
	\vspace{-0.5ex}
	where $n = - {N_{\rm cp}}, \ldots , N- 1$. 
	
	\subsubsection{Demodulation at Bob}

	After transmission over the channel with $P$ paths, where the channel coefficient, time delay, and Doppler shift of the $i$-th path are denoted by ${h_i^{\rm B}}$, ${\tau _i^{\rm B}}$, ${f_{d,i}^{\rm B}}$, respectively, the received time-domain signal vector at Bob is given by
	\begingroup
	\setlength{\abovedisplayskip}{3pt}
	\setlength{\belowdisplayskip}{3pt}
	\begin{equation}\label{eq:received_sig_Bob_time}
		{r}_{\rm B}\left[ n \right] = \sum\limits_{i = 1}^P {{{\tilde h}^{\rm B}_i}} {s}_A\left[ {n - {l_i^{\rm B}}} \right]{e^{j2\pi {f_i^{\rm B}}n}} + {{w}}_{\rm B}\left[ n \right],
	\end{equation}
	\endgroup
	where ${{{\tilde h}^{\rm B}_i}} = {h_i^{\rm B}}{e^{ - j2\pi {f^{\rm B}_{d,i}}{\tau ^{\rm B}_i}}}$,  ${l^{\rm B}_i} {=} {{{\tau^{\rm B}_i}} \mathord{\left/{\vphantom {{{\tau ^{\rm B}_i}} {{t_s}}}} \right.
			\kern-\nulldelimiterspace} {{t_s}}}$, ${f^B_i} {=} {f^B_{d,i}}{t_{\rm s}}$ with ${t_{\rm s}}$ being the sampling interval, and $\mathbf{w}_{{\rm B}}\in {{\mathbb{C}}^{N\times 1}}$ is an AWGN vector with a power spectral density $\sigma _{n,{\rm B}}^2$.

	After serial to parallel conversion (S/P) and CPP removal, the received SE-AFDM signal at Bob in the time domain is given by 
	\begingroup
	\setlength{\abovedisplayskip}{3pt}
	\setlength{\belowdisplayskip}{3pt}
	\begin{equation}
		{\mathbf{r}}_{{\rm B}} = \sum\limits_{i = 1}^P {{{\tilde h}_i^{\rm B}}{{\bf{\Gamma }}_{{\rm cpp}{_i}}}{{\bf{\Delta }}_{{f_{i}^{\rm B}}}}{{\bf{\Pi }}^{{l_i^{\rm B}}}}{\bf{s}}_{\rm A}}  + {\bf{w}}_{{\rm B}}.
		\label{eq: relation between R and x in matrix form} 
	\end{equation}
	\endgroup
		
	Then, multiplying ${\mathbf{r}}_{{\rm B}}$ by matrix ${\bm{\Lambda} _{{c_1}}}$ and performing DFT, we obtain
	\vspace{-0.2cm}
	\begin{eqnarray}   
		\begin{aligned} \label{eq:ES_AFDM_Rx_TD}
			\begin{split}
				&{\mathbf{r}}'_{{\rm B}}\hspace{-1mm}=\hspace{-1mm} \sum\limits_{i = 1}^P {{{\tilde h}^{\rm B}_i} \mathbf{F}{\bm{\Lambda} _{{c_1}}} {{\bf{\Gamma }}_{{\rm cpp}{_i}}}{{\bf{\Delta }}_{{f^{\rm B}_{i}}}}{{\mathbf{\Pi }}^{{l^{\rm B}_i}}}{\bm{\Lambda} ^{H}_{{c_1}}} \mathbf{F}^{ H}{\bm{\Lambda} ^{H}_{{c_2,{\rm A}}}}\mathbf{x} }  + {\mathbf{w}}'_{{\rm B}},\hspace{-1mm}
			\end{split}
		\end{aligned}
	\end{eqnarray}
	where ${\bf{w}}'_{{\rm B}} = \mathbf{F}{\bm{\Lambda} _{{c_1}}} {\bf{w}}_{{\rm B}}$. 
	
	After that, ${\mathbf{r}}'_{{\rm B}}$ is multiplied by the dynamic matrix ${\bm{\Lambda} _{{c_2,{\rm B}}}} \!= \! {\rm diag}\left( {{e^{ - j2\pi {c}^{\rm B}_{2,\mu}\left[m\right]{m^2}}},m  \!= \! 0, \ldots ,N {-} 1} \right)$, where $\mathbf{c}^{\rm B}_{2,\mu}$ is an ${N \times 1}$ parameter vector corresponding to the $\mu$-th AFDM symbol for Bob.
	Each entry in $\mathbf{c}^{\rm B}_{2,\mu}$ is selected from the codebook $\mathcal{C}_2$ according to an index controlled by a ${\log _2}M$-bit random sequence.
	Leveraging the prior knowledge of the LPPN sequence generator polynomials of Alice, Bob synchronizes the local generator to produce the random sequence. The synchronization details will be elaborated later. Now, the received matrix at Bob in the affine domain can be written in matrix form as
	\begingroup
	\setlength{\abovedisplayskip}{2pt}
	\setlength{\belowdisplayskip}{2pt}
	\begin{align} \label{eq:ES_AFDM_Rx_AD}
		\mathbf{y}_{{\rm B}}&= \sum\limits_{i = 1}^P {{{\tilde h}^{\rm B}_i}{\bm{\Lambda} _{{c_2,{\rm B}}}} \mathbf{F}{\bm{\Lambda} _{{c_1}}} {{\bf{\Gamma }}_{{\rm cpp}{_i}}}{{\bf{\Delta }}_{{f^{\rm B}_{i}}}}{{\bf{\Pi }}^{{l^{\rm B}_i}}}{\bm {\Lambda} ^{H}_{{c_1}}} \mathbf{F}^{H}{\bm{\Lambda} ^{ H}_{{c_2,{\rm A}}}}\mathbf{x} }  + {\mathbf{\bar w}}_{{\rm B}} \nonumber \\
	& = \sum\limits_{i = 1}^P {{{\tilde h}^{\rm B}_i}{\bm{\Lambda} _{{c_2,{\rm B}}}} \mathbf{H}_i^0 {\bm{\Lambda} ^{ H}_{{c_2,{\rm A}}}}\mathbf{x} }  + {\bf{\bar w}}_{B} \nonumber \\
		& = \mathbf{H}_{{\rm eff},{\rm B}} \mathbf{x}  + {\bf{\bar w}}_{{\rm B}},
	\end{align} 
	\endgroup
	where ${\bf{\bar w}}_{{\rm B}} = {\bm{\Lambda} _{{c_2,{\rm B}}}}\mathbf{F}{\bm{\Lambda} _{{c_1}}} {\bf{w}}_{{\rm B}}$. 
	
	
	\vspace{-1ex}
	
	\subsection{Input-Output Relation of SE-AFDM Between Alice and Bob}
	
	\vspace{0.5ex}
	\vspace{-1ex}
	
	Based on \eqref{eq:ES_AFDM_Rx_AD}, one has 
	\begingroup
	\setlength{\abovedisplayskip}{4pt}
	\setlength{\belowdisplayskip}{4pt}
	\begin{eqnarray}
		{H}_{i}^0[p,q]=\frac1Ne^{j\frac{2\pi}N\left(Nc_1(l^{\rm B}_i)^2-ql^{\rm B}_i\right)}{\mathcal{F}}_{i, \rm B}[p,q],
	\end{eqnarray} 
	\endgroup
	%
	%
	%
where ${{\mathcal{F}}_{i,{\rm B}}}\left[ {p,q} \right] {=} \frac{{{e^{ - j2\pi \left( {p - q - {\nu^{\rm B}_i} + 2N{c_1}{l^{\rm B}_i}} \right)}} - 1}}{{{e^{ - j\frac{{2\pi }}{N}\left( {p - q - {\nu^{\rm B}_i} + 2N{c_1}{l^{\rm B}_i}} \right)}} - 1}}$ with ${\nu ^{\rm B}_i} = N{f^{\rm B}_i}$. The effective channel $\mathbf{H}_{{\rm eff},{\rm B}}$ can be rewritten as
\begingroup
\setlength{\abovedisplayskip}{2pt}
\setlength{\belowdisplayskip}{2pt}
\begin{align} \label{eq:Heff_B_pq}
	&{H}_{{\rm eff},{\rm B}} \left[p,q\right]  = \sum\limits_{i = 1}^P {{{\tilde h}^{\rm B}_i} { H}_i^0[p,q]{e^{j2\pi \left[ {{{{c}}_{2,\mu}^{\rm A}}[q]{q^2} - {{{c}}_{2,\mu}^{\rm B}}[p]{p^2}} \right]}}} \nonumber \\  
	&\!= \hspace{-1.7mm} \sum \limits_{i = 1}^P \hspace{-1.7mm} \frac1N  \hspace{-0.5mm}{{\tilde h}^{\rm B}_i}\hspace{-0.5mm} {e^{j2\pi \!\left[ {{{{c}}_{2,\mu}^{\rm A}}\![q]{q^2} \!\hspace{-0.3mm} - \hspace{-0.3mm} {{{c}}_{2,\mu}^{\rm B}}\![p]{p^2}}\! \right]}}\hspace{-1.3mm}\times\hspace{-1.3mm} {e^{j\!\frac{{2\pi }}{N}\!\left(\! {N\!{c_1}\!(l^{\rm B}_i)\!^2 \hspace{-0.3mm}\!-\hspace{-0.3mm} q{l^{\rm B}_i}}\! \right)}}\hspace{-1mm}{\mathcal{F}_{i,{\rm B}}}\![p,\!q].
\end{align}
\endgroup
	Thus, the input-output relation can be expressed as \eqref{eq:in_out_rel_Bob} on the top of this page. We can see from \eqref{eq:in_out_rel_Bob} that the vector ${\mathbf{c}}_{2,\mu}^{\rm A}$ generated at Alice  affects every received symbol at Bob.

	
	
	\begin{figure*}[htbp]
	\begin{align} \label{eq:in_out_rel_Bob}
	y_{\rm B}[p]\!\! =\! \!\sum\limits_{i = 1}^P {\tilde h_i^{\rm B}} \!\sum\limits_{q = 0}^{N - 1} \!{{H_{i,{\rm B}}}} [p,q]x[q]\!\! +\! \!w[p] \!=\! \!\sum\limits_{i = 1}^P {\tilde h_i^{\rm B}} {e^{ - j2\pi c_{2,\mu}^{\rm B}[p]{p^2}}}\sum\limits_{q = 0}^{N - 1} {\frac{1}{N}{e^{j2\pi c_{2,\mu}^{\rm A}[q]{q^2}}}} {e^{j\frac{{2\pi }}{N}\left( {N{c_1}{{\left( {l_i^{\rm B}} \right)}^2} - ql_i^{\rm B}} \right)}}\!{{\cal F}_{i,{\rm B}}}[p,q]x[q] \!\!+ \!\! w[p].
	\end{align}  
	\hrule
\end{figure*}

	Since the vectors ${\bf{c}}_{2,\mu}^{\rm A}$ and ${\bf{c}}_{2,\mu}^{\rm B}$ are controlled by the LPPN sequences of Alice and Bob, respectively, synchronizing these sequences yields ${{\bf{c}}_{2,\mu}^{\rm A}}$ and ${{\bf{c}}_{2,\mu}^{\rm B}}$, i.e., ${{{c}}_{2,\mu}^{\rm A}}[q] = {{{c}}_{2,\mu}^{\rm B}}[q],$ $q = 0, \ldots N-1$.
	Thus, the effect of  ${{\bf{c}}_{2,\mu}^{\rm A}}$ can be eliminated, and Bob can detect  $\mathbf{x}$ in the minimum mean square error (MMSE) criterion as \cite{jiang2011performance}
	\begingroup
	\setlength{\abovedisplayskip}{3.5pt}
	\setlength{\belowdisplayskip}{3.5pt}
	\begin{eqnarray} \label{Bob_MMSE}
		{\hat {\bf{x}}}_{\rm B}  = \mathbf{H}_{{\rm eff},{\rm B}}^H{\left( {{\mathbf{H}_{{\rm eff},{\rm B}}}\mathbf{H}_{{\rm eff},{\rm B}}^H + \sigma _{n,{\rm B}}^2{{\bf{I}}_N}} \right)^{ - 1}}{{\bf{y}}_{\rm B}}.
	\end{eqnarray}
	\endgroup
	
	It is shown that Bob can recover the transmitted symbols $\bf{x}$ from the received ${\bf{y}}_{{\rm B}}$ after the synchronization of the two LPPN sequence generators of Alice and Bob. The synchronization strategy between Alice and Bob will be presented in Sec. \ref{sec:syn_strategy}. Simulation and experimental results will verify this conclusion in Sec. \ref{sec:results}.

	\vspace{-1ex}
	
	\section{Security Analyses of SE-AFDM System} \label{sec:security}
	
	In this section, we analyze the security of the proposed SE-AFDM system. Specifically, the input-output relation of SE-AFDM between Alice and Eve is derived. Building upon this relation, we reveal that the effect of the vector ${\bf c}^{\rm A}_{2,\mu}$ cannot be eliminated at Eve. Moreover, the effective SINR of Eve is analyzed. 
	
		\vspace{-1ex}
	\subsection{Input-Output Relation of SE-AFDM at Eve}
	It is assumed that Eve has a very strong capability, that is, it knows the fixed waveform parameters, e.g., $c_1$, $N$, ${N_{\rm cp}}$ and the codebook $\mathcal{C}_2$. However, Eve does not know the configuration parameters of the LPPN sequence generator of Alice, which means that Eve is unable to reconstruct and synchronize ${\bf c}^{\rm A}_{2,\mu}$.

	
	The received SE-AFDM signal at Eve in the time domain can be expressed as 
	\begin{equation} 
		{\bf{r}}_{{\rm E}} = \sum\limits_{i = 1}^P {{{\tilde h}_i^{\rm E}}{{\bf{\Gamma }}_{{\rm cpp}_i}}{{\bf{\Delta }}_{{f_{i}^{\rm E}}}}{{\bf{\Pi }}^{{l_i^{\rm E}}}}{\bf{s}}_{{\rm A}}}  + {\bf{w}}_{{\rm E}} ,
		\label{eq: relation between R and x in matrix form_Eve}
	\end{equation}
	where $\mathbf{w}_{{\rm E}}\in {{\mathbb{C}}^{N\times 1}}$ is an AWGN vector with power spectral density $\sigma _{n,{\rm E}}^2$.  
	

	Similar to (\ref{eq:ES_AFDM_Rx_TD}), after performing S/P and discarding CPP, followed by multiplication with the matrix ${\bm{\Lambda} _{{c_1}}}$ and a DFT operation, we obtain
	\begin{equation} \label{eq:ES_AFDM_Eve}
		{\bf{r}}'_{{\rm E}}= \sum\limits_{i = 1}^P {{{\tilde h}_i^{\rm E}} {\bf F}{\bm{\Lambda} _{{c_1}}} {{\bf{\Gamma }}_{{\rm cpp}{_i}}}{{\bf{\Delta }}_{{f^{\rm E}_{i}}}}{{\bf{\Pi }}^{{l^{\rm E}_i}}}{{\bf \Lambda} ^{H}_{{c_1}}} {\bf F}^{ H}{{\bf \Lambda} ^{ H}_{{c_2,{\rm A}}}}{\bf x} }  + {\bf{w}}'_{{\rm E}},
	\end{equation}
	where ${\bf{w}}'_{{\rm E}} = {\bf F}{{\bf \Lambda} _{{c_1}}} {\bf{w}}_{{\rm E}}$.

	Then, ${\bf{r}}'_{{\rm E}}$ is multiplied by matrix ${\bm{\Lambda} _{{c_2,{\rm E}}}}  {=}  {\rm diag} \hspace{-1mm}\left( \hspace{-1mm} {{e^{ - j2\pi{c}^{\rm E}_{2,\mu}\left[m\right]{m^2}}}\hspace{-1mm} ,\hspace{-0.5mm} m\hspace{-1mm} =\hspace{-1mm} 0, \hspace{-0.5mm} \ldots , \hspace{-0.5mm} N {-} 1} \hspace{-1mm}\right)$, where ${\bf c}^{\rm E}_{2,\mu}$ is the parameter vector corresponding to the $\mu$-th AFDM symbol for Eve. The received matrix at Eve in the affine domain can be written in matrix form as

	\vspace{-1ex}
	
	\begingroup
	\setlength{\abovedisplayskip}{3pt}
	\setlength{\belowdisplayskip}{3pt}
	\begin{align} \label{eq:ES_AFDM_Rx_EVE}
		{\bf y}_{{\rm E}}&= \sum\limits_{i = 1}^P {{{\tilde h}^{\rm E}_i}{{\bf \Lambda} _{{c_2,{\rm E}}}} {\bf F}{\bm{\Lambda} _{{c_1}}} {{\bf{\Gamma }}_{{\rm cpp}{_i}}}{{\bf{\Delta }}_{{f^{\rm E}_{i}}}}{{\bf{\Pi }}^{{l^{\rm E}_i}}}{\bm{\Lambda} ^{ H}_{{c_1}}} {\bf F}^{ H}{\bm{\Lambda} ^{ H}_{{c_2,{\rm A}}}}{\bf x} }  + {\bf{\bar w}}_{{\rm E}} \nonumber \\
		& = {\bf H}_{{\rm eff},{\rm E}}' {\bf x}'  + {\bf{\bar w}}_{{\rm E}},
	\end{align}
	\endgroup
	where ${\bf H}_{{\rm eff},{\rm E}}' \hspace{-0.5ex}= \hspace{-0.5ex}\sum\limits_{i = 1}^P {{{\tilde h}^{\rm E}_i}{\bm{\Lambda} _{{c_2,{\rm E}}}} {\bf F}{\bm{\Lambda} _{{c_1}}} {{\bf{\Gamma }}_{{\rm cpp}{_i}}}{{\bf{\Delta }}_{{f^{\rm E}_{i}}}}{{\bf{\Pi }}^{{l^{\rm E}_i}}}\hspace{-0.5ex}{\bm{\Lambda} ^{ H}_{{c_1}}} {\bf F}^{H}} $,\hspace{-0.1ex} ${{\bf{x}}' }\hspace{-0.2ex} =\hspace{-0.2ex} {\bm{\Lambda} ^{ H}_{{c_2,{\rm A}}}}{\bf x} $, ${\bm{\Lambda} ^{H}_{{c_2,{\rm A}}}} {=} {\rm diag}\left( {{e^{ j2\pi {c}^{\rm A}_{2,\mu}\left[q\right]{q^2}}},q = 0, \ldots ,N {-} 1} \right)$, and ${\bf{\bar w}}_{{\rm E}} = {\bm{\Lambda} _{{c_2,{\rm E}}}}{\bf F}{\bm{\Lambda} _{{c_1}}} {\bf{w}}_{{\rm E}}$.

	

		\vspace{-1ex}
	\subsection{Analyzing of the Effect of ${\bf c}^{\rm A}_2$ on Eve}
	
	
	
	Eve is assumed to know the matrix ${\bf H}_{{\rm eff},{\rm E}}'$. Hence, the vector ${{\bf{x}}^\prime }$ can be estimated by Eve using MMSE as \cite{jiang2011performance}
	\begin{eqnarray}
		{\bf{\hat x}}_{\rm E}^\prime  = {\bf H}_{{\rm eff},{\rm E}}'^H{\left( {{{\bf H}_{{\rm eff},{\rm E}}'}{\bf H}_{{\rm eff},{\rm E}}'^H + \sigma _{n,{\rm E}}^2{{\bf{I}}_N}} \right)^{ - 1}}{{{\bf y}}_{\rm E}},
	\end{eqnarray}
where ${\bf{\hat x}}_{\rm E}^\prime$ is the information-symbol vector $\bf{x}$ affected by the diagonal matrix ${\bm{\Lambda} ^{ H}_{{c_2,{\rm A}}}}$ and the residual noise, denoted by ${\bf{\tilde w}}_{\rm E}$.

Neglecting the impact of residual noise, we have
	\begin{eqnarray}\label{eq:c2_x}
		\left\{\begin{array}{rl}
			e^{j2\pi {c}^{\rm A}_{2,\mu}\left[0\right]0^2} \cdot x[0]  &=\hat{{x}}_{\rm E}^{'}[0]\\[-4pt]
			&\vdots\\[-4pt]
			e^{j2\pi {c}^{\rm A}_{2,\mu}\left[N-1\right](N-1)^2} \cdot  x[N-1] &=\hat{{x}}_{\rm E}^{'}[N-1].
		\end{array} \right.
	\end{eqnarray}
	From (\ref{eq:c2_x}), it can be seen that the received $\hat{{x}}_{\rm E}^{'}[q]$ consists of both the transmitted information symbol ${x}\left[q\right]$ and ${c}^{\rm A}_{2,\mu}\left[q\right]$ for $q = 0, \ldots N-1$. If both ${\bf x}$ and ${\bf c}^{\rm A}_{2,\mu}$ are varying and unknown to Eve, Eve cannot recover ${\bf x}$ and ${\bf c}^{\rm A}_{2,\mu}$ from the received $\hat{\mathbf{x}}_{\rm E}^{'}$, since each equation has two unknowns when $q\geq1$. Meanwhile, for different AFDM symbols, ${x}\left[q\right]$ can be designed as pilots when $q=0$.
	

		\vspace{-1ex}
	\subsection{Analysis of Effective SINR of Eve}
	
	
	
	This section analyzes the security of the SE-AFDM system using the effective SINR of Eve as the metric.
	When the effective SINR of Eve decreases, less information is eavesdropped due to the reduced BER at Eve. 

	In AWGN channel, the output SINR at Bob can be expressed as
	\begingroup
	\setlength{\abovedisplayskip}{3pt}
	\setlength{\belowdisplayskip}{3pt} 
	\begin{eqnarray}\label{eq:SINR_B}
		\mathrm{SINR}_{\rm B} = \frac{{{p_s} \alpha ^2_{\rm B} }}{{\sigma _{n,{\rm B}}^2}} ={\gamma} _{\rm B},
	\end{eqnarray}
	\endgroup
	where ${p_s}$ is the transmit power of Alice, $\alpha_{\rm B}$ represents the large-scale fading from Alice to Bob, and $\gamma_{\rm B} = {p_s} \alpha ^2_{\rm B}/{{\sigma _{n,{\rm B}}^2}}$ denotes the output signal-to-noise ratio (SNR) of the received signal at Bob.
	


	At Eve, the estimation of the $q$-th symbol can be written as
	\begin{align}
		{\hat{x}}_{\rm E}^{\prime}\left[q\right]& =  {{x}}_{\rm E}^{\prime}\left[q\right] + {\tilde w}_{\rm E}\left[q\right] = {x}[q]e^{j2\pi{c}^{\rm A}_{2,\mu}[q]q^2}+{\tilde w}_{\rm E}\left[q\right] \nonumber\\
		&={x}[q]+(e^{j2\pi{c}^{\rm A}_{2,\mu}[q]q^2}-1){x}\left[q\right]+{\tilde w}_{\rm E}\left[q\right],
	\end{align} 
	where $q = 0, \ldots N-1$, ${\tilde w}_{\rm E}[q]$ denotes the residual noise after symbol detection.

	Therefore, the effective output SINR of the $q$-th symbol at Eve after signal processing is given by [\mycitenum{zou2007compensation} Eq. (16)]
	\begin{align} \label{eq:Eq25_1}
		{{\mathop{\rm SINR}\nolimits} _{{\rm E},q}}\hspace{-0.5ex} = \hspace{-0.5ex}\frac{{ \mathbb{E}\left\{ {|x[q]{|^2}} \right\}}}{\mathbb{E}{\left\{ {|x[q]{|^2}} \right\}\mathbb{E}\left\{ {{{\left| {{e^{j2\pi c_{2,\mu}^{\rm A}[q]{q^2}}} \hspace{-0.5ex}- \hspace{-0.5ex}1} \right|}^2}} \right\} \hspace{-0.5ex}+\hspace{-0.5ex} \sigma _{n,{\rm E}}^2}},
	\end{align}
	where ${c}^{\rm A}_{2,\mu}[q]$ is randomly selected at Alice from the codebook $\mathcal{C}_2$ based on a truncated LPPN sequence.

\begin{proposition}  \label{proposition1}
	After theoretical derivation, (\ref{eq:Eq25_1}) can be rewritten as
\vspace{-0.5em}
	\begin{align}\label{eq:SINR_E_local}
		&{{\mathop{\rm SINR}\nolimits} _{{\rm{E}},q}} =\nonumber\\ 
		& \left\{ {\begin{array}{*{20}{l}}
				\hspace{-2ex} {{\gamma _{\rm{E}}},}&\hspace{-1.8ex}{{\Delta _{\rm{A}}}{q^2} \in \mathbb{Z},}\\
				\hspace{-2ex} {\frac{{{\gamma _{\rm{E}}}}}{{{\gamma _{\rm{E}}}\!\left\{ {\!2 -\! \frac{2}{M}\!\Re \left\{\!\! {{e\!^{ - j2\pi {c_{2,\max }}{q^2}}}\!\frac{{{e^{j2\pi {\Delta _{\rm{A}}}{q^2}M}}\! -\! 1}}{{{e^{j2\pi {\Delta _{\rm{A}}}{q^2}}} \!- \!1}}\! } \right\} } \!\!\right\}\! +\! 1}}}\!,&\hspace{-1.8ex}{{\text{otherwise}}{{,}}}
		\end{array}} \right.
	\end{align}
	where $\gamma_{\rm E} = {p_s} \alpha ^2_{\rm E}/{{\sigma _{n,{\rm E}}^2}}$ denotes the output SNR of the received signal at Eve, $\alpha_{\rm E}$ represents the large-scale fading from Alice to Eve, and ${\Delta_{\rm A}}=\frac{{2{c_{2,\max }}}}{M-1}$ denotes the codebook interval.
\end{proposition}	
\begin{proof}  
	See Appendix \ref{proofcorollary1}.
\end{proof}  

	The average effective SINR of $N$ received symbols at Eve can be calculated as follows:     
    \begingroup
	\setlength{\abovedisplayskip}{3pt}
	\setlength{\belowdisplayskip}{3pt} 
	\begin{align}\label{eq:SINR_E}
		{{\mathop{\rm SINR}\nolimits} _{{\rm E}}} = \frac{1}{N}\!\!\sum\limits_{q = 0}^{N - 1} {{\mathop{\rm SINR}\nolimits} _{{\rm E},q}}. 
	\end{align}  
	\endgroup

	{\bf Discussion on the effective SINR of Eve:} From (\ref{eq:SINR_E_local}) and (\ref{eq:SINR_E}), it can be observed that the effective SINR of Eve in the SE-AFDM system is affected by $c_{2,{\rm max}}$ and $M$. 

	
	For finite $M$ and $P$, and as $c_{2,{\rm max}}$ tends to zero, 
	i.e., $c_{2,{\rm max}} \to 0$, ${e^{ - j2\pi {q^2}{c_{2,\max }}}} \approx 1$,
	$ {e^{j\frac{{4\pi {q^2}{c_{2,\max }}M}}{{M - 1}}}} -1 \approx \frac{{j4\pi {q^2}{c_{2,\max }}M}}{{M - 1}}$, $ {e^{j\frac{{4\pi {q^2}{c_{2,\max }}}}{{M - 1}}}}-1 \approx \frac{{j4\pi {q^2}{c_{2,\max }}}}{{M - 1}}$, and 
	${\frac{1}{M}\Re \left\{ {{e^{ - j2\pi {c_{2,\max }}{q^2}}}\frac{{{e^{j2\pi {\Delta _{\rm{A}}}{q^2}M}} - 1}}{{{e^{j2\pi {\Delta _{\rm{A}}}{q^2}}} - 1}}} \right\}}=1$, one has $\mathrm{SINR}_{{\rm E}} = \gamma _{\rm E}$ = ${p_s} \alpha ^2_{\rm E}/{{\sigma _{n,{\rm E}}^2}}$. When the transmission power ${p_s}$ at Alice increases, the SNR of Bob $\gamma_{\rm B}$ and the SNR of Eve $\gamma_{\rm E}$ rise. In this case, there is a high risk of eavesdropping, indicating a lack of security. 

	\vspace{-5pt}
	\begin{example}
		Considering $\gamma _{\rm E}$ $= 25$ dB, $N = 1024$, and $M = 10^{5}$, the effective SINR at Eve in the SE-AFDM system with different values of $c_{2,\max}$ is illustrated in Fig.~\ref{fg:cs_vs_c2max}. It is shown that the effective SINR of Eve declines as $c_{2,{\rm max}}$ increases. Consistent with the theoretical analysis,  $\mathrm{SINR}_{{\rm E}} = \gamma _{\rm E}$ $= 25$dB when $c_{2,{\rm max}}$ is small. As $c_{2,\rm max}$ continues to increase, the effective SINR eventually approaches -0.93 dB. In summary, only a sufficiently large $c_{2,{\rm max}}$ can provide security in the SE-AFDM system.
	\end{example}

	

	
	\begin{figure}[htbp]
		\vspace{-5pt}
		\centering
		\includegraphics[width=2.5in]{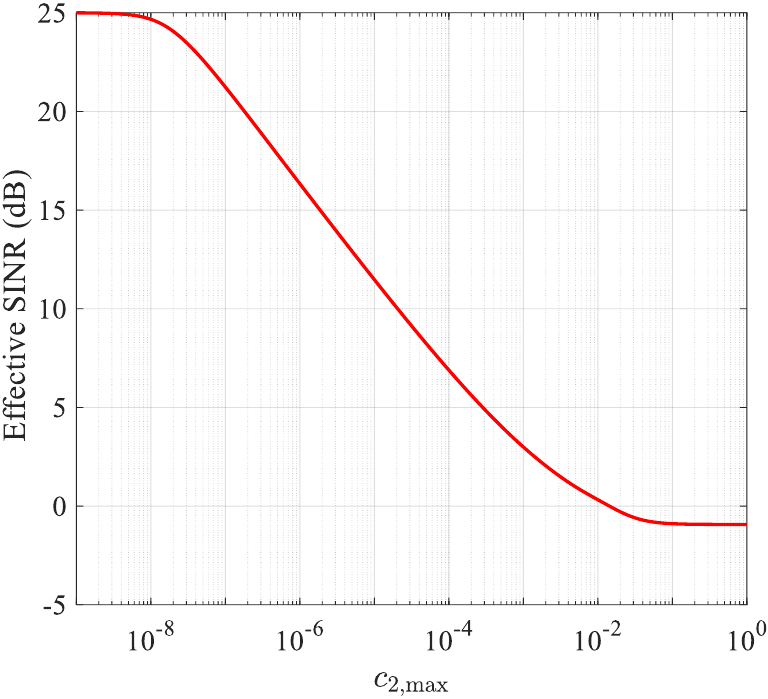}  
		\caption{The effective SINR at Eve versus $c_{2,\rm max}$ of SE-AFDM with $\gamma _{\rm E}$ = 25 dB.  
			\label{fg:cs_vs_c2max}}
		\vspace{-5pt}
	\end{figure}
	


	

	
	
	We now briefly analyze the complexity of brute-force search at Eve and the spectrum efficiency of the proposed SE-AFDM system.
	If Eve attempts to recover ${\bf c}^{\rm A}_{2,\mu}$ of the $\mu$-th AFDM symbol by brute force, the search space size is $M^N$ for the proposed SE-AFDM system with parameter-domain spreading. For comparison, the search space size for a data-domain DSSS system using an $N$-length LPPN sequence is $2^N$. Let ${\Delta_{\rm E}}$ denote the fixed search interval adopted by Eve.
    Since the codebook $\mathcal{C}_2$ is publicly available, Eve may adopt a fixed search interval ${\Delta _{\rm{E}}}$ larger than the codebook interval ${\Delta_{\rm A}}$ to reduce the search space size. Assume that the codebook size satisfies $M \ge 2$ and the search interval is ${\Delta _{\rm{E}}} = u{\Delta_{\rm A}}$, where $u = 1, \ldots ,M-1$. Then  the search values are $A_{{\rm E},k} = - {c_{2,\max }} + k{\Delta _{\rm E}}$, $k=0,\ldots,{\left\lfloor {\frac{{M - 1}}{u}} \right\rfloor}$, and the new search space size is ${\left( {\left\lfloor {\frac{{M - 1}}{u}} \right\rfloor \! \!+ \!\!1} \right)\!^N}$. Denote the set of search values as ${{{\bf{A}}_{\rm{E}}}}$.
    Each ${ c}^{\rm{E}}_{2,\mu}[q]$ of ${ \bf c}^{\rm{E}}_{2,\mu}$ is the closest to ${c}^{\rm{A}}_{2,\mu}[q]$ among the set of search values ${{{\bf{A}}_{\rm{E}}}}$, i.e., 
    \begingroup
    \setlength{\abovedisplayskip}{3pt}
    \setlength{\belowdisplayskip}{3pt}
     \begin{eqnarray}  \label{choose_c2E}
    c_{2,\mu}^{\rm{E}}[q] = \arg \mathop {\min }\limits_{{A_{{\rm{E}},k}} \in {{\bf{A}}_{\rm{E}}}} \left| {{A_{{\rm{E}},k}} - c_{2,\mu}^{\rm{A}}[q]} \right|. 
     \end{eqnarray} 
     \endgroup
    At Eve, the estimation error can be modeled as $\delta_{c_2}\! =\! \left| {c_{2,\mu}^{\rm{E}}[q]\! -\! c_{2,\mu}^{\rm{A}}[q]} \right|$, where ${\delta _{{c_2}}} \!\!\in\!\! \{\xi{\Delta_{\rm A}}  \! \mid\!\xi \!\in\! \left\{ {0, \ldots ,{\xi _{\max }}} \right\}\}$ and ${\xi _{\max }} = \max \left\{ {M - 1 - u\left\lfloor {\frac{{M - 1}}{u}} \right\rfloor ,\left\lfloor {\frac{u}{2}} \right\rfloor } \right\}$.	In terms of spectrum efficiency, the proposed SE-AFDM system achieves the same spectrum efficiency as the existing AFDM system.

	
	
	


	\section{Synchronization Framework of SE-AFDM} \label{sec:syn_strategy}
	
	
	
	
	
	In this section, we present a synchronization framework to achieve the synchronization of the dynamic parameter $c_2$ at Bob. To enable synchronization in fast time-varying channels, we design an affine-domain frame structure for the SE-AFDM system. Based on this, a corresponding synchronization strategy is proposed, where Bob achieves $c_2$ synchronization by synchronizing the LPPN sequence.
	
	

	
	
		\vspace{-2ex}

	\subsection{Proposed Frame Structure for SE-AFDM System}


   
   As shown in Fig. \ref{fg:frame_structure}, the proposed frame structure consists of $K$ AFDM symbols, organized into three blocks: the frame synchronization block, the LPPN sequence synchronization block and the secure data transmission block. Accordingly, the $K$ AFDM symbols can be arranged in matrix form as 
    \begingroup
   \setlength{\abovedisplayskip}{2pt}
   \setlength{\belowdisplayskip}{3.5pt}
   \begin{eqnarray}\label{frame_S}
   	{\bf{S}} = \left[ {{{\bf{S}}_{{\rm{head}}}},{{\bf{S}}_{{\rm{LPPN}}}},{{\bf{S}}_{{\rm{com}}}}} \right],
   \end{eqnarray} 
   \endgroup
    where ${\bf{S}}\in {\mathbb{C}^{ N \times K}}$, $N=2Q+L+1$, ${{{\bf{S}}_{{\rm{head}}}} \in {\mathbb{C}^{N \times J}}}$ represents the frame synchronization block using fixed  ${\mathbf{c}}^{\rm A}_{2,\mu}$ vectors with $\mu\!=\!0,\ldots,J-1$ for frame synchronization, ${{{\bf{S}}_{{\rm{LPPN}}}} \in {\mathbb{C}^{N \times {\left( {E - J} \right)}}}}$ denotes the LPPN sequence synchronization block using fixed  ${\mathbf{c}}^{\rm A}_{2,\mu}$ vectors with $\mu=J,\ldots,E-1$ to enable LPPN sequence synchronization between Alice and Bob, and ${{{\bf{S}}_{{\rm{com}}}} \in {\mathbb{C}^{N \times {\left( {K-E} \right)}}}}$ corresponds to the secure data transmission block using dynamic ${\mathbf{c}}^{\rm A}_{2,\mu}$ vectors with $\mu=E,\ldots,K-1$ for secure transmission of random communication data.
    
   

   \begin{figure}[htbp]  
   	\vspace{-5pt}
   	\centering
   	\includegraphics[width=3.5in]{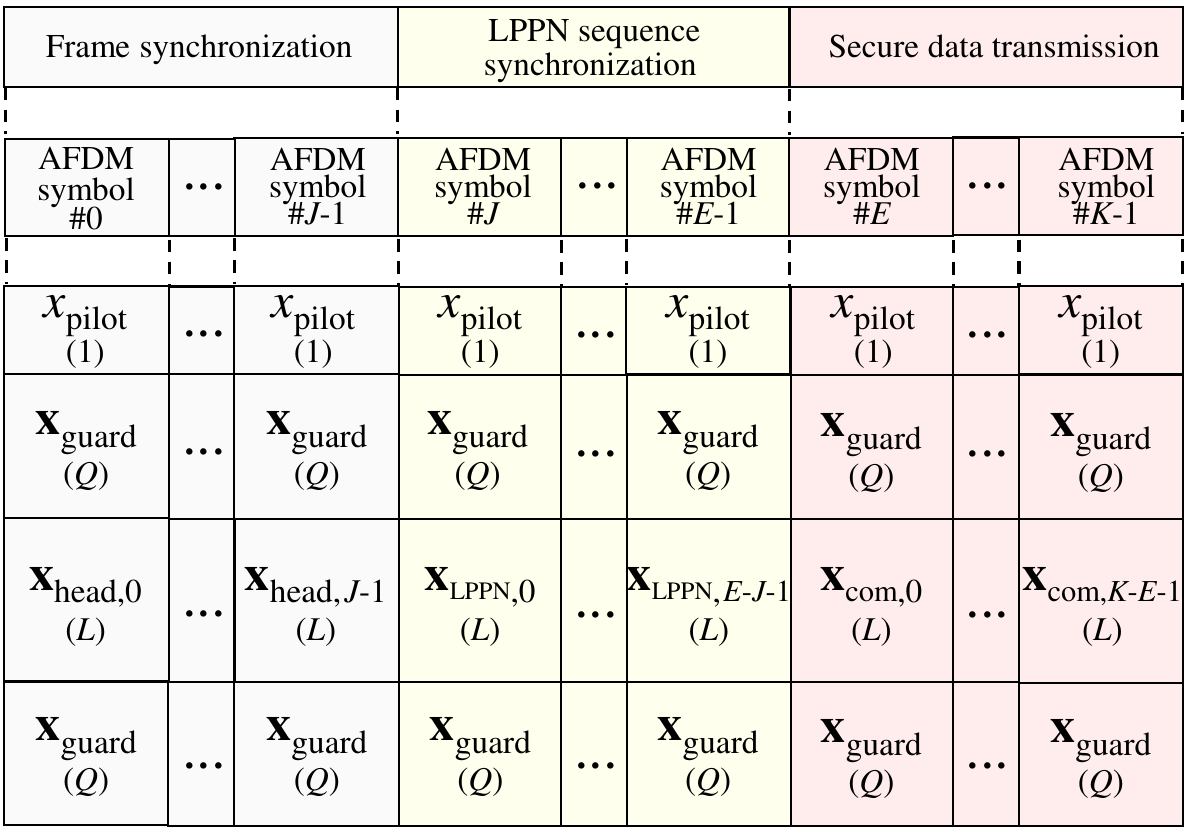}
   	\caption{ Affine-domain frame structure designed for the synchronization of the SE-AFDM system.
   		\label{fg:frame_structure}     }
   	\vspace{-5pt}
   \end{figure}

   \subsubsection{Frame synchronization block} 
   
   
\indent    For the frame synchronization block, ${{\bf{S}}_{{\rm{head}}}}$ is given by
%
   \begingroup
\setlength{\abovedisplayskip}{3pt}
\setlength{\belowdisplayskip}{3.5pt}
    \begin{eqnarray}  
	{{\bf{S}}_{{\rm{head}}}} = \left[ {{{\bf{s}}_{{\rm{head,}}0}}, \ldots ,{{\bf{s}}_{{\rm{head, }}J - 1}}} \right],
	\end{eqnarray}
	\endgroup
	where ${{{\bf{s}}_{{\rm{head,}}i}}} \in {\mathbb{C}^{N \times 1}}$  denotes an AFDM symbol in the affine domain for $i = 0, \ldots, J-1$, which is defined as 
   \begingroup
\setlength{\abovedisplayskip}{3pt}
\setlength{\belowdisplayskip}{3.5pt}
    \begin{eqnarray}
	{{\bf{s}}_{{\rm{head,}}i}} = {\left[ {{{ {{x^T_{{\rm{pilot }}}}} }},{{ {{{\bf{x}}^T_{{\rm{guard }}}}} }},{{{{{\bf{x}}^T_{{\rm{head,}}i}}} }},{{ {{{\bf{x}}^T_{{\rm{guard }}}}} }}} \right]^T},
    \end{eqnarray}  
    \endgroup 
   where ${x_{{\rm{pilot }}}}$ represents a pilot symbol for channel estimation, ${{\bf{x}}_{{\rm{guard }}}} = \mathbf{0}_{Q \times 1}$ denotes ${Q}$ guard intervals, and ${{{\bf{x}}_{{\rm{head,}}i}}}\in \mathbb{C}^{L \times 1}$  contains $R$-QAM symbols for $i = 0, \ldots, J-1$.  
   
   For ${{{\bf{x}}_{{\rm{head,}}i}}}$, we first generate a pseudo-noise (PN) sequence $\mathbf{m} \in \{0,1\}^{JL {\log _2}R \times 1}$ using an m-sequence generator \cite{tian2009m}. 
   Then, the sequence $\mathbf{m}$ is divided into $J$ contiguous segments to form a matrix ${{\bf{M}}} = \left[ {{\bf{m}}_{0}}, \ldots ,{{\bf{m}}_{J-1}} \right] \in {\left\{ {0,1} \right\}^{L{{\log }_2}R \times J}}$.  By applying $R$-QAM modulation to the vectors ${{\bf{m}}_i}, i=0,\ldots,J-1$, we obtain the corresponding complex symbol vectors ${{{\bf{x}}_{{\rm{head,}}i}}} \in \mathbb{C}^{L \times 1}$ for $i=0, \ldots, J-1$. Specifically, the $k$-th element of ${{{\bf{x}}_{{\rm{head,}}i}}}$ is given by
   \begin{eqnarray}
   x_{{{\rm{head,}}i}}[k]\!\!=\!\!\mathcal{M}\!_{R\text{-QAM}}\!\left(\!m_{i}[k{\log _2}R :(\!k\!\!+\!\!1\!){\log _2}R\!\!-\!\!1] \right)\!, 
   \end{eqnarray} 
   where $ k\! =0, \ldots,\! L-1$. For $i\!=\!0,\ldots,J-1$, we construct ${{{\bf{x}}_{{\rm{head,}}i}}}$ into a matrix $\mathbf{X}_{\mathrm{head}}\!=\!\big[\mathbf{x}_{\mathrm{head},0},\ldots,\mathbf{x}_{\mathrm{head},J-1}\big] \in {\mathbb{C}^{L \times J}}$.




   From (\ref{eq:c2_x}), we can see that the received signal is affected by ${\mathbf{c}}^{\rm A}_{2,\mu}$ and the information symbols. To eliminate the impact of ${\mathbf{c}}^{\rm A}_{2,\mu}$ on the frame header, ${\mathbf{c}}^{\rm A}_{2,\mu}$ is fixed to be a known constant vector $\mathbf{u}_{N \times 1}$ shared by Alice and Bob. The modulated time-domain signal of the AFDM symbol ${{{\bf{s}}_{{\rm{head,}}i}}}$ is given by 
   \begin{equation}
   	\hspace*{-0.5em}
   	{s}_{{\rm A	1}}\left[ n \right]\!=\!\frac{1}{\sqrt{N}}\sum\limits_{m=0}^{N-1}{{{{s}}_{{\rm{head,}}i}} }\left[ m \right]{{e}^{j2\pi \left( {{c}_{1}}{{n}^{2}}+{c}^{\rm A}_{2,\mu}\left[m\right]{{m}^{2}}+\frac{mn}{N} \right)}} ,
   \end{equation}
   where $n \!= \!- {N_{\rm cp}}, \ldots , N-1$, $i = 0, \ldots, J-1$ and $c_{2,\mu}^{\rm{A}}[m] \!= \!u$.







   
    \subsubsection{LPPN sequence synchronization block} 
    
  \indent  The LPPN sequence synchronization block ${{\bf{s}}_{{\rm{LPPN}}}}$, composed of $E-J$ AFDM symbols in the affine domain, is given by
     \begingroup
  \setlength{\abovedisplayskip}{3pt}
  \setlength{\belowdisplayskip}{3.5pt}
   \begin{eqnarray}  
   	{{\bf{S}}_{{\rm{LPPN}}}} = \left[ {{{\bf{s}}_{{\rm{LPPN,}}0}}, \ldots ,{{\bf{s}}_{{\rm{LPPN, }}E-J - 1}}} \right],
   \end{eqnarray}
   \endgroup
   where ${{{\bf{s}}_{{\rm{LPPN,}}i}}} \in {\mathbb{C}^{N \times 1}}$, $i = 0, \ldots, E-J-1$. The ${{{\bf{s}}_{{\rm{LPPN,}}i}}}$ is formed by  
    \begin{eqnarray}
   	{{\bf{s}}_{{\rm{LPPN,}}i}} = {\left[ {{{ {{x^T_{{\rm{pilot }}}}} }},{{ {{{\bf{x}}^T_{{\rm{guard }}}}} }},{{ {{{\bf{x}}^T_{{\rm{LPPN,}}i}}} }},{{ {{{\bf{x}}^T_{{\rm{guard }}}}} }}} \right]^T},
   \end{eqnarray}
   where ${{{\bf{x}}_{{\rm{LPPN,}}i}}}\in \mathbb{C}^{L \times 1}$ and $i = 0, \ldots, E-J-1$.

   
   To achieve the LPPN sequence synchronization between Alice and Bob, the state parameters of the LPPN sequence generator at Alice are transmitted in the LPPN sequence synchronization block. These include the counter values ${n_{{\rm{X1A}}}}$, ${n_{{\rm{X1B}}}}$, ${n_{{\rm{X2A}}}}$, ${n_{{\rm{X2B}}}}$, ${n_{{\rm{X1}}}}$, ${n_{{\rm{X2}}}}$ and the current register states ${\bf{s}}^k_{\rm{X1A}}$, ${\bf{s}}^k_{\rm{X1B}}$, ${\bf{s}}^k_{\rm{X2A}}$, ${\bf{s}}^k_{\rm{X2B}}$. First, the decimal counter values are converted to binary and concatenated to construct a vector ${\bf{n}} $. Then, we have  
   \begin{eqnarray}
   	{\bf{w}} = \left[ {{{\bf{n}}}^T,\left({\bf{s}}^k_{\rm{X1A}}\right)\!^{T}\!,\left({\bf{s}}^k_{\rm{X1B}}\right)\!^{T}\!,\left({\bf{s}}^k_{\rm{X2A}}\right)\!^{T}\!,\left({\bf{s}}^k_{\rm{X2B}}\right)\!^{T}} \right]^T,
   \end{eqnarray} 
   where ${\bf{w}} \!\!\in \!\!{\left\{\! {0,1} \!\right\}\!\!^{{D}\! \times \!1}}$, ${\bf{n}} \!\!\in \!\!{\left\{\! {0,1} \!\right\}}\!^{(D-4S) \!\times\! 1}$, ${\bf{s}}^k_{\beta} \!\!\in \!\!{\left\{\! {0,1} \!\right\}}\!^{S \!\times\! 1}$, 
  $\beta\!\! \in\!\! \{ {\!\text{X1A,X1B,X2A,X2B}\!}\}$.
   Given the importance of ${\bf{w}}$, DSSS with a spreading factor $F$ is employed to transmit ${\bf{w}}$ under low-SNR conditions. Let ${\bf{m}}_{F}$ denote the $F \times 1$ spreading sequence, which can be generated as described in \cite{tian2009m}. After converting $\bf{w}$ and ${\bf{m}}_{F}$ from unipolar to bipolar by multiplying each bit by 2 and subtracting 1, we obtain ${{\bf{w}}^b} \!\!\in \!\!{\left\{\! {-1,1} \!\right\}\!\!^{{D}\! \times \!1}}$ and ${\bf{m}}^b_{F}\!\! \in\!\! {\left\{ {-1,1} \right\}\!^{F \!\times\! 1}}$. The DSSS output is then obtained as follows \cite{alagil2019randomized}: 
      \begingroup
   \setlength{\abovedisplayskip}{3pt}
   \setlength{\belowdisplayskip}{3.5pt}
\begin{eqnarray}
	{\bf{x}}^b = {{\bf{w}}^b} \otimes {\bf{m}}_F^b,
\end{eqnarray}
\endgroup
 where ${\bf{x}}^b\in{\left\{ {-1,1} \right\}^{{D}F \times 1}}$. 
 By applying bipolar-to-unipolar conversion and zero-padding to ${\bf{x}}^b $, the vector $\mathbf{x} \in {\left\{\! {0,1} \!\right\}\!\!^{(E-J)L{\log _2}R}}$ is obtained.
 The vector $\mathbf{x}$ is then reorganized into a matrix ${{{\bf X}}}$, defined as ${{{\bf X}}} = \left[ {{{\bf x}}_{0}}, \ldots ,{{{ \bf x}}_{E-J-1}} \right]$. By applying $R$-QAM modulation to ${{{\bf x}}_{i}} \in {\left\{\! {0,1} \!\right\}\!\!^{L{\log _2}R}}$, $i=0,\ldots,E-J-1$, we obtain the symbol vectors ${{{\bf{x}}_{{\rm{LPPN,}}i}}} \in \mathbb{C}^{L \times 1}$. The $k$-th element of ${{{\bf{x}}_{{\rm{LPPN,}}i}}}$ is given by
 \begin{eqnarray}
 	x_{{{\rm{LPPN,}}i}}[k]\!\!=\!\!\mathcal{M}\!_{R\text{-QAM}}\!\left(\!x_{i}[k{\log _2}R :(\!k\!\!+\!\!1\!){\log _2}R\!\!-\!\!1] \right)\!, 
 \end{eqnarray} 
 where $ k\! =0, \ldots,\! L-1$. The collection of all ${{{\bf{x}}_{{\rm{LPPN,}}i}}}$  yields a matrix ${{\bf{X}}_{{\rm{LPPN}}}} = \left[ {{{\bf{x}}_{{\rm{LPPN,}}0}}, \ldots ,{{\bf{x}}_{{\rm{LPPN, }}E-J - 1}}} \right]  \in {\mathbb{C}^{L \times (E-J)}} .$

 Here, ${\mathbf{c}}^{\rm A}_{2,\mu}$ for $\mu=J,\ldots,E-1$ is predefined as $\mathbf{u}_{N \times 1}$, which is pre-shared and known a priori by Alice and Bob. For ${{{\bf{s}}_{{\rm{LPPN,}}i}}}$,  the resulting time-domain signal is
 \begingroup
 \setlength{\abovedisplayskip}{3.5pt}
 \setlength{\belowdisplayskip}{3pt}
 \begin{equation}
 	\hspace*{-0.7em}
 	{s}_{{\rm A2}}\left[ n \right] \!=\!\!\frac{1}{\sqrt{N}}\!\sum\limits_{m=0}^{N-1}{{{{s}}_{{\rm{LPPN,}}i}} }\left[ m \right]{{e}^{j2\pi \left( {{c}_{1}}{{n}^{2}}+{c}^{\rm A}_{2,\mu}\left[m\right]{{m}^{2}}+\frac{mn}{N} \right)}} ,
 \end{equation}
 \endgroup
 where $n \!= \!- {N_{\rm cp}}, \ldots , N-1$, $i \!=\!0, \ldots, E\!-\!J\!-\!1$, and $c_{2,\mu}^{\rm{A}}[m] \!= \!u$.

 
%
%
%
%
%
%

    \subsubsection{Secure data transmission block} 
   
     \indent  To transmit random communication data, we employ the secure data transmission block across $K-E$ AFDM symbols in the affine domain, which is constructed as
     \begingroup
     \setlength{\abovedisplayskip}{3pt}
     \setlength{\belowdisplayskip}{3.5pt}
    \begin{eqnarray}  
    	{{\bf{S}}_{{\rm{com}}}} = \left[ {{{\bf{s}}_{{\rm{com,}}0}}, \ldots ,{{\bf{s}}_{{\rm{com, }}K-E - 1}}} \right],
    \end{eqnarray}
    \endgroup
     where ${{{\bf{s}}_{{\rm{com,}}i}}} \in {\mathbb{C}^{N \times 1}}$, $i = 0, \ldots, K-E-1$. The ${{{\bf{s}}_{{\rm{com,}}i}}}$ is given by
     \begingroup
     \setlength{\abovedisplayskip}{3pt}
     \setlength{\belowdisplayskip}{3.5pt}
     \begin{eqnarray}
      	{{\bf{s}}_{{\rm{com,}}i}} = {\left[ {{{ {{x^T_{{\rm{pilot }}}}} }},{{ {{{\bf{x}}^T_{{\rm{guard }}}}} }},{{{{{\bf{x}}^T_{{\rm{com,}}i}}} }},{{ {{{\bf{x}}^T_{{\rm{guard }}}}} }}} \right]^T},
     \end{eqnarray}
      \endgroup
     where ${{{\bf{x}}_{{\rm{com,}}i}}}\!\in \!\mathbb{C}^{L \times 1}$ denotes the symbol vector obtained from $L \log_2 R$ random information bits through $R$-QAM modulation.

     


     
     To enhance the security of the SE-AFDM system, ${\mathbf{c}}^{\rm A}_{2,\mu}$ is set to be dynamic for $\mu=E,\ldots,K-1$. Each  ${c}^{\rm A}_{2,\mu}\left[m\right]$ is generated from the codebook $\mathcal{C}_2$ using the $\log _{2} M$-bit truncated LPPN sequence, as introduced in Section~\ref{sec:section3}. The resulting time-domain signal is    
     \begingroup
     \setlength{\abovedisplayskip}{2pt}
     \setlength{\belowdisplayskip}{3pt}
     \begin{equation}
     	\hspace*{-0.5em}
     	{s}_{{\rm A3}}\left[ n \right]\!=\!\frac{1}{\sqrt{N}}\sum\limits_{m=0}^{N-1}{{{{{s}}_{{\rm{com,}}i}}}}\left[ m \right]{{e}^{j2\pi \left( {{c}_{1}}{{n}^{2}}+{c}^{\rm A}_{2,\mu}\left[m\right]{{m}^{2}}+\frac{mn}{N} \right)}} ,
     \end{equation} 
      \endgroup
     where $n = - {N_{\rm cp}}, \ldots , N-1$ and $i = 0, \ldots, K-E-1$.

	\vspace{-2ex}
	
	\subsection{Synchronization Strategy} 
	
    Based on the frame structure, the synchronization strategy can be organized into the following three stages. The channel estimation method follows \cite{bemani2023affine}, while the equalization method is based on \eqref{Bob_MMSE}.
    
	
	\vspace{0.5ex}
	\noindent\emph{Stage 1: Frame Detection} 
	\vspace{0.5ex}
	
    The first stage is to detect the frame start position using a sliding window mechanism \cite{hou2022multisignal}. Once the sliding window aligns with the frame synchronization block, a distinct correlation peak emerges, enabling reliable frame header detection. As illustrated in Fig. \ref{fg:Sliding_time_window}, a signal segment of length $J \times \left( {{N_{\rm cp}} + N} \right)$ is extracted as the window slides by ${N_{\rm cp}}$ symbols. After channel estimation and equalization with $\mathbf{c}^{\rm B}_{2,\mu} = \mathbf{u}_{N \times 1}$ for $\mu\!=\!0,\ldots,J-1$, Bob obtains the received complex matrix ${{\bf{X}}'_{{\rm{head}}}} = \left[ {{{\bf{x}}'_{{\rm{head,}}0}}, \ldots ,{{\bf{x}}'_{{\rm{head, }}J - 1}}} \right]$. Subsequent $R$-QAM demodulation and symbol detection yield the bit matrix ${{\bf{M}}'} = \left[ {{\bf{m}}'_{0}}, \ldots ,{{\bf{m}}'_{J-1}} \right]$. The received PN sequence $\mathbf{m}'$ is given by
   \begin{equation}
{\mathbf{m}^{\prime}=\operatorname{vec}\left(\mathbf{M}^{\prime}\right)},
\end{equation}
    where $\mathbf{m}^{\prime}\in {\left\{ {0,1} \right\}^{JL{{\log }_2}R \times 1}}$. Then, Bob correlates the local PN sequence ${\bf{m}}$ with the received $\mathbf{m}'$ and detects the position of the frame synchronization block based on a correlation peak. The correlation peak is obtained by
    \begingroup
    \setlength{\abovedisplayskip}{1pt}
    \setlength{\belowdisplayskip}{1pt}
    \begin{equation} \label{eq:correlation}
    	{{{r}}_{\rm{m}}} = {\bf{m}}' \cdot {{\bf{m}}^T}.
    \end{equation}  
    \endgroup
    Once the correlation peak surpasses a predefined threshold, the start position of the frame is detected, enabling the subsequent procedure. The threshold is set based on the SNR of the received signal and the PN sequence length.


    

    \begin{figure}[htbp]  
    	\vspace{-5pt}
    	\centering
    	\includegraphics[width=3.5in]{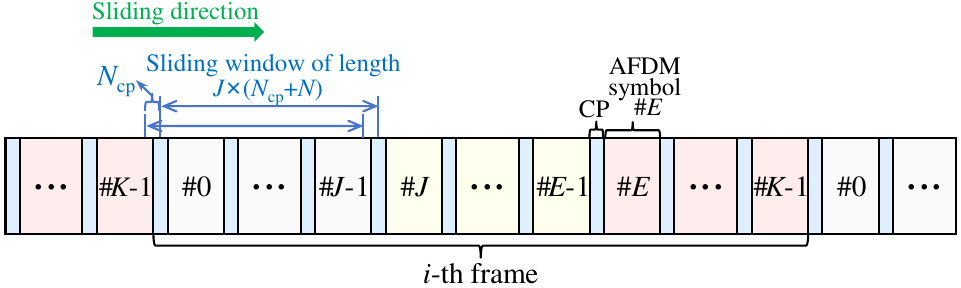}
    	\caption{The sliding window mechanism for the frame header detection.
    		\label{fg:Sliding_time_window}}
    	\vspace{-5pt}
    \end{figure}

	\vspace{0.5ex}
	\noindent\emph{Stage 2: LPPN sequence  synchronization}
	\vspace{0.5ex}

	 Stage $2$ synchronizes the LPPN sequences of Alice and Bob, thereby enabling the subsequent secure data transmission. By applying channel estimation and MMSE equalization using $\mathbf{c}^B_{2,\mu} = \mathbf{u}_{N \times 1}$ for $\mu\!=\!J,\ldots,E-1$, Bob obtains a complex matrix ${\bf{X}}'_{{\rm{LPPN}}} \in {\mathbb{C}^{L \times (E-J)}}$.
	 After vectorizing ${\bf{X}}'_{{\rm{LPPN}}}$, performing $R$-QAM demodulation, and removing the zero padding, the information vector ${\left( {{{\bf{x}}^b}} \right)'} \in \mathbb{R}^{  DF \times 1}$ is obtained. Then, Bob reshapes ${\left( {{{\bf{x}}^b}} \right)'}$ into a matrix
	 $\mathbf{S} \in \mathbb{R}^{  D \times F}$. With a local spreading sequence identical to that of Alice, despreading is performed by
	  \begingroup
	 \setlength{\abovedisplayskip}{3pt}
	 \setlength{\belowdisplayskip}{3pt}
	  \begin{equation} \label{eq:despread}
	 {\left( {{{\bf{w}}^b}} \right)'} = \frac{1}{F}  \mathbf{S}  {\bf{m}}_F^b,
	 \end{equation}
	   \endgroup
	 where ${\left( {{{\bf{w}}^b}} \right)'}\!\!\in \!\!{ \mathbb{R}^{{D}\! \times \!1}}$. After bit detection, ${\bf{w}}'$ is obtained.
	Subsequently, Bob initializes the local LPPN sequence generator with ${\bf{w}}'$ to synchronize with the LPPN sequence of Alice. Thus, $c_2$ synchronization can be achieved at Bob by the synchronized LPPN sequence. Notably, even if Eve obtains ${\bf{w}}'$, synchronization between the LPPN sequences of Eve and Alice remains unachievable due to the unknown configuration parameters of the LPPN sequence generator at Alice, thereby resulting in the failure of $c_2$ synchronization.

	\vspace{0.5ex}
	\noindent\emph{Stage 3: Secure data reception}  
	\vspace{0.5ex}
	
	Once the LPPN sequences between Alice and Bob have been synchronized in Stage 2, Bob proceeds to detect the data within the secure data transmission block.	
	After performing S/P, discarding CPP, multiplying by ${\bm{\Lambda} _{{c_1}}}$, and performing DFT, the resulting signal can be written as   
			  \begingroup
		\setlength{\abovedisplayskip}{3pt}
		\setlength{\belowdisplayskip}{3pt}
				\begin{align} 
							{\mathbf{r}}'_{\rm B}\hspace{-1mm}&=\hspace{-1mm}\! \sum\limits_{i = 1}^P \!{{{\tilde h}^{\rm B}_i} \mathbf{F}\!{\bm{\Lambda} _{{c_1}}}\! {{\bf{\Gamma }}_{{\rm cpp}{_i}}}{{\bf{\Delta }}_{{f^{\rm B}_{i}}}}{{\mathbf{\Pi }}^{{l^{\rm B}_i}}}\!\!{\bm{\Lambda} ^{H}_{{c_1}}} \mathbf{F}^{ H} \!\mathbf{x}' } \!\! +\!\! {\mathbf{w}}'_{B}, \nonumber\\ 
							& = {\mathbf{ H}}_{{\rm eff},{\rm B}}' \mathbf{x}'	+ {\mathbf{w}}'_{\rm B},	
					\end{align}
					\endgroup
		where ${\bf H}_{{\rm eff},{\rm B}}' \hspace{-0.8ex}= \hspace{-0.8ex}\sum\limits_{i = 1}^P {{{\tilde h}^{\rm B}_i}  {\bf F}{\bm{\Lambda} _{{c_1}}}\! {{\bf{\Gamma }}_{{\rm cpp}{_i}}}{{\bf{\Delta }}_{{f^{\rm B}_{i}}}}{{\bf{\Pi }}^{{l^{\rm B}_i}}}\hspace{-0.5ex}{\bm{\Lambda} ^{ H}_{{c_1}}} {\bf F}^{H}}$, ${\mathbf{x}}'\hspace{-0.7ex} = \hspace{-0.7ex}{\bm{\Lambda} ^{ H}_{{c_2,\rm A}}}{{\bf x}_{{\rm com},k}} $, and $k = 0, \ldots, K-E-1$.
		
	%
	
%

	Then, Bob estimates the channel parameters, including the number of propagation paths $P$, the channel coefficient ${{\tilde h}^{\rm B}_i}$, the time delay ${{l}^{\rm B}_i}$, and the Doppler shift ${f^{\rm B}_i}$, for $i = 1, \ldots, P$. Accordingly, the efficient channel matrix is given by
	\begingroup
	\setlength{\abovedisplayskip}{3pt}
	\setlength{\belowdisplayskip}{3pt}
	   \begin{align} 
		{{ H}}'_{{\rm eff},{\rm B}} \left[p,q\right] = \hspace{-1mm} \sum \limits_{i = 1}^P \hspace{-1mm} \frac1N  \hspace{-0.5mm}{{\tilde h}^{\rm B}_i}\hspace{-0.5mm} \hspace{-1mm}\times\hspace{-1mm} {e^{j\frac{{2\pi }}{N}\left( {N{c_1}(l^{\rm B}_i)^2 \hspace{-0.3mm}-\hspace{-0.3mm} q{l^{\rm B}_i}} \right)}}\hspace{-1mm}{\mathcal{F}_{i,{\rm B}}}[p,q].
	\end{align}
	\endgroup
	After equalization, Bob obtains ${\hat {\bf{x}}}'_{\rm B} = {\bm{\Lambda} ^{ H}_{{c_2,{\rm A} }}}{{ {\bf{x}}_{{\rm com},k }}}+ {\tilde {\bf {w}}}_{\rm B}.$
	For $\mu\!=\!E,\ldots,K-1$, Bob eliminates the impact of ${\mathbf{c}}_{2,\mu}^{\rm A}$ by constructing ${\mathbf{c}}^{\rm B}_{2,\mu}$ synchronized with ${\mathbf{c}}^{\rm A}_{2,\mu}$, i.e., ${\bm{\Lambda}_{{c_2,{\rm B}}}} = {\bm{\Lambda}_{{c_2,{\rm A}}}}$. The detailed procedure for constructing ${\mathbf{c}}^{\rm B}_{2,\mu}$ is provided in Section~\ref{sec:section3}. By multiplying ${\bm{\Lambda} _{{c_2,{\rm B}}}}$, Bob gets 
	\begingroup
	\setlength{\abovedisplayskip}{3pt}
	\setlength{\belowdisplayskip}{3pt}
	\begin{eqnarray}\label{x_final}
	{\hat {\bf{x}}}_{\rm B} = {{ {\bf{x}}_{{\rm com},k }}}+ {\bm{\Lambda}}_{{c_2,{\rm B}}}{\tilde {\bf {w}}}_{\rm B}.
	\end{eqnarray} 
	\endgroup
	From \eqref{x_final}, it can be seen that received  ${\hat {\bf{x}}}_{\rm B}$ is no longer affected by the difference between $\mathbf{c}^{\rm B}_{2,\mu}$ and $\mathbf{c}^{\rm A}_{2,\mu}$, allowing recovery of the transmitted data.

	\section{Simulation and Experimental Results} \label{sec:results}
	In this section, both simulation results and experimental results are presented to validate the performance of the proposed SE-AFDM system. Inspired by the P-code used in the GPS system \cite{IS-GPS}, the configuration parameters of the LPPN sequence generator are listed as follows. For the four shift registers, the stage number is $S=12$, the polynomial coefficients of the four shift registers are ${{\bf{e}}_{{\rm{X1A}}}}\!\!=\!\! [0,\!0,\!0,\!0,\!0,\!1,\!0,\!1,\!0,\!0,\!1,\!1]^T$,
	${{\bf{e}}_{{\rm{X1B}}}}\!\!=\!\! [ 1,\! 1,\! 0,\! 0, \!1,\! 0, \!0,\! 1,\! 1,\! 1,\! 1,\! 1]^T$, ${{\bf{e}}_{{\rm{X2A}}}}\!\!=\!\! [1,\!0,\!1,\!1,\!1,\!0,\!1,\!1,\!1,\!1,\!1,\!1]^T$, ${{\bf{e}}_{{\rm{X2B}}}}\!\!=\!\! [0, \!1, \!1, \!1, \!0,\! 0,\!0,\! 1,\! 1, \!0,\! 0,\! 1]^T$, and the corresponding initial states are ${{\bf{s}}^0_{{\rm{X1A}}}}\!\! =\!\! [0,\!0,\!0,\!1,\!0,\!0,\!1,\!0,\!0,\!1,\!0,\!0]^T$, ${{\bf{s}}^0_{{\rm{X1B}}}}\!\! =\!\! [0,\!0,\!1,\!0,\!1,\!0,\!1,\!0,\!1,\!0,\!1,\!0]^T$, ${{\bf{s}}^0_{{\rm{X2A}}}}\!\! =\!\! [1,\!0,\!1,\!0,\!0,\!1,\!0,\!0,\!1,\!0,\!0,\!1]^T$, ${{\bf{s}}^0_{{\rm{X2B}}}}\!\! =\!\! [0,\!0,\!1,\!0,\!1,\!0,\!1,\!0,\!1,\!0,\!1,\!0]^T$.
	The shortened cycles of $\mathrm{X1A}$, $\mathrm{X1B}$, $\mathrm{X2A}$ and $\mathrm{X2B}$ sequences are given as $T_{\rm X1A} = 4092$, $T_{\rm X1B} = 4093$, $T_{\rm X2A} = 4092$, and  $T_{\rm X2B} = 4093$. The count threshold of the $\mathrm{X1A}$ cycle, $\mathrm{X2A}$ cycle and $\mathrm{X1}$ cycle are set as $\theta_{\rm X1A}=3750$, $\theta_{\rm X2A}=3750$ and $\theta_{\rm X1}=15345037$. The cycle of $\mathrm{X2}$ sequence exceeds that of the $\mathrm{X1}$ sequence by $d=37$. Moreover, the spreading factor is $F=15$, and the spreading sequence is ${\bf{m}}_F=[1,0,0,0,1,1,1,1,0,1,0,1,1,0,0]^T$.

	\subsection{Simulation Results}   
	
    In the simulation,  quadrature phase shift keying (QPSK) symbols are transmitted. The maximum integer part of the normalized Doppler shift is $\alpha _{\rm max} = 2$, corresponding to a maximum speed of 1350 km/h \cite{bemani2023affine}. For $P=3$ paths, each path has a different Doppler shift generated by the Jakes model, i.e., $\nu _i = \alpha _{\rm max} \cos(\theta _i)$, where $\theta _i$ is uniformly distributed over $\left[-\pi,\pi\right]$ \cite{bemani2023affine}. The complex gain of the $i$-th path ${{h}_{i}}$ is set to be independent complex Gaussian random variables with zero mean and $1/P$ variance. Additionally, ${ \bf c}^{\rm B}_{2,\mu} = {\bf c}^{\rm A}_{2,\mu}$ for Bob, and ${ \bf c}^{\rm E}_{2,\mu} = {\bf 0}_{N \times 1}$ for Eve.
    Unless otherwise specified,  the simulation parameters are listed in Table \ref{tab:table2}.



    \renewcommand{\arraystretch}{1}  
    \begin{table}[h]
    	\caption{Simulation parameters\label{tab:table2} }  
    	\centering
    	\begin{tabular}{|c|m{5cm}|c|}
    		\hline 
    		\textbf{Symbol} &\textbf{Parameter}  &\textbf{Value}\\ 
    		\hline 
    		$f_c$ &Carrier frequency &24 GHz\\
    		\hline
    		$B$ &Bandwidth &15.36 MHz\\   
    		\hline
    		$\Delta f$ &Subcarrier spacing &15 kHz\\
    		\hline
    		$N$ &Number of subcarriers &1024\\
    		\hline 
    		$N_{\rm cp}$ &Number of CPP &17\\
    		\hline 
    		$M$ &Size of the codebook &1024\\
    		\hline
    		$N_{\rm sym}$ &Number of AFDM symbols per frame &1\\
    		\hline 	
    		$P$ &Number of channel paths &3\\
    		\hline
    		$l$ &Delay taps &[0,1,2]\\
    		\hline 
    		$\alpha _{\rm max}$ & Maximum integer part of the normalized Doppler shift &2\\
    		\hline
    	\end{tabular}
    \end{table}

	The BER performances versus SNR with different $c_{2,\rm max}$ are shown in Fig. \ref{fg:BER_vs_SNR_diff_c2}. In our proposed SE-AFDM system, the BER performances at Bob are almost the same as those of the existing AFDM system for any $c_{2,\rm max}$. However, as $c_{2,\rm max}$ increases, the BER performance at Eve deteriorates significantly, approaching 0.5. These BER results show that the security performance of the SE-AFDM system improves as $c_{2,\rm max}$ increases, which is consistent with the observation in Fig.~\ref{fg:cs_vs_c2max}.
	
	
	


			
			\begin{figure}[htbp]
				\centering
				\vspace{-5pt}
				\includegraphics[width=2.5in]{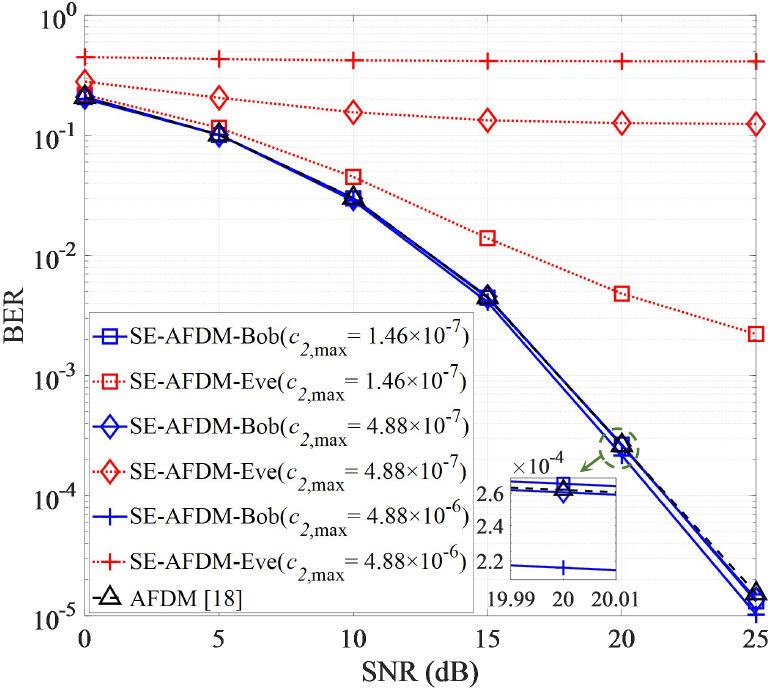}
				\caption{The BER performances versus SNR with different $c_{2,\rm max}$ of our SE-AFDM and the existing AFDM \cite{bemani2023affine}.
					\label{fg:BER_vs_SNR_diff_c2}}
				\vspace{-5pt}
			\end{figure}
			\vspace{-2pt}
			
			Next, the impact of channel estimation errors on the BER performance of the SE-AFDM system is evaluated using the channel estimation method in \cite{bemani2023affine}. The results are shown in Fig. \ref{fg:var_doppler}.	
			We set $c_{2,\rm max} = 4.88\times 10^{-6}$ and the pilot-symbol SNR to $30$ dB, i.e., SNR$\rm_p = 30$ dB.
			With estimated CSI, the BER of the proposed SE-AFDM system at Bob coincides with that of the AFDM system in \cite{bemani2023affine}, both of which are slightly worse than the BER of the perfect CSI case. Meanwhile, the BER of the SE-AFDM system at Eve is about 0.5 with estimated CSI, indicating that the SE-AFDM system remains effective under practical channel estimation.
			

		Then, we investigate the impact of the search interval $\Delta_{\rm E}$ on BER performance at Eve with $c_{2,\rm max}$ = $4.88\times 10^{-5}$, $M = 10^{6}$ and codebook interval of $9.76\times10^{-11}$. Each ${ c}^{\rm E}_{2,\mu}[q]$ is chosen as the element closest to ${c}^{\rm A}_{2,\mu}[q]$  among the search values, as shown in \eqref{choose_c2E}. As illustrated in Fig. \ref{fg:bias}, the BER at Eve degrades as $\Delta_{\rm E}$ increases. The BER at Eve is larger than 0.1 when $\Delta_{\rm E} > 7.8\times10^{-7}$. When $\Delta_{\rm E} < 9.77\times10^{-8}$, the BER drops below $1.77\times10^{-5}$. In this case, Eve needs to search about 1000 times for each ${c}^{\rm E}_{2,\mu}[q]$ to ensure the accuracy of the received data. These findings provide valuable insight into the design of the codebook $\mathcal{C}_2$: specifically, system security can be enhanced by maximizing the codebook range.	
		

			\begin{figure}[t]
			\vspace{-5pt}
			\centering
			\includegraphics[width=2.5in]{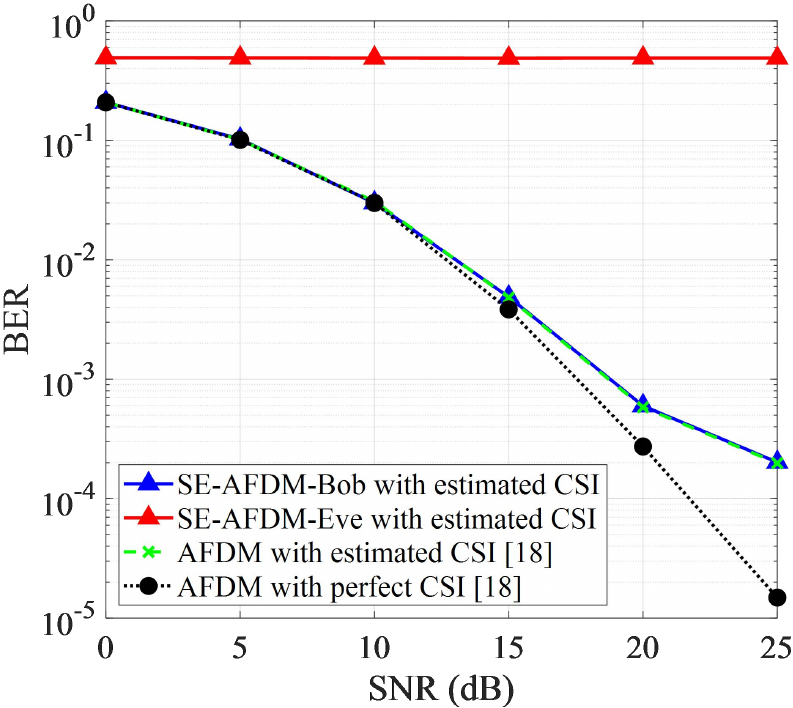}
			\caption{The BER performances of SE-AFDM and AFDM with estimated CSI at SNR$\rm_p$ = 30dB.
					\label{fg:var_doppler}}
				\vspace{-5pt}
			\end{figure}

				\begin{figure}[t]
					\centering
					\vspace{-5pt}
					\includegraphics[width=2.5in]{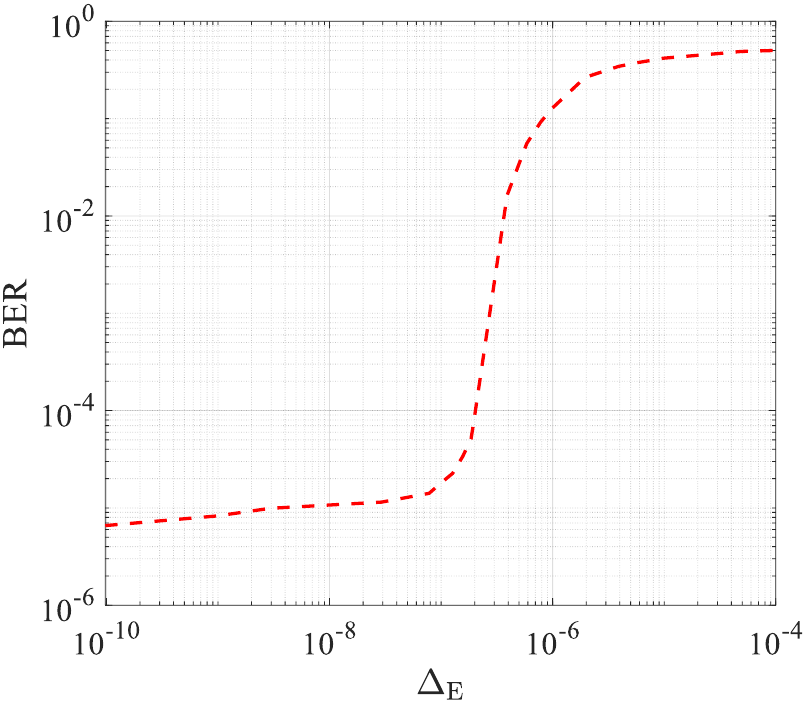}
					\caption{The BER performance of Eve versus the search interval $\Delta_{\rm E}$ with SNR = 25 dB.
							\label{fg:bias}}
						\vspace{-5pt}
				\end{figure}

				\begin{figure}[H]
					\centering
					\vspace{-5pt}
					\includegraphics[width=2.5in]{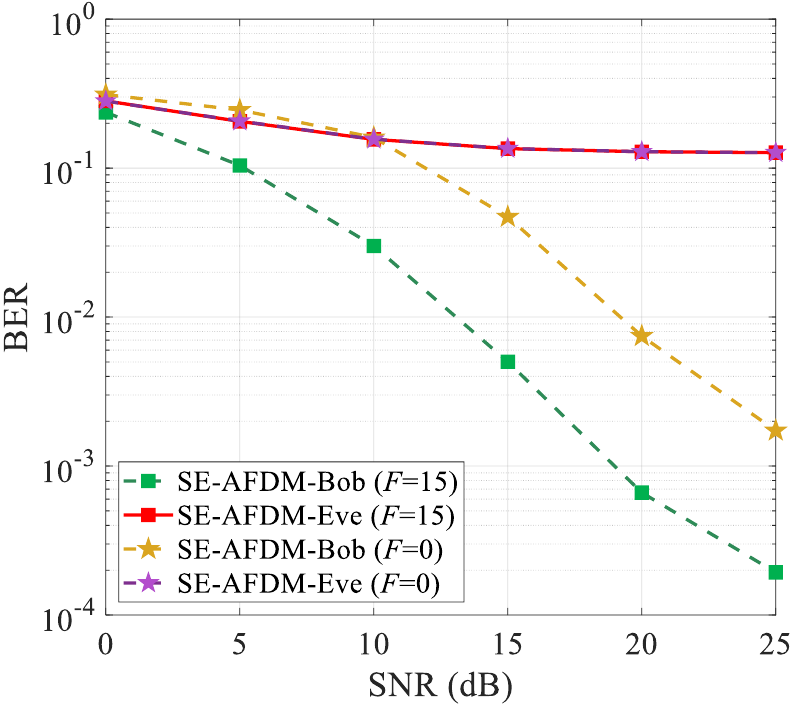}
					\caption{The BER performances of SE-AFDM with spreading factor $F=15$ and $F = 0$ under estimated CSI.
							\label{fg:DSSS_simu}}
						\vspace{-5pt}
					\end{figure}

					
					 
					 Finally, the impact of applying DSSS to the state parameters of the LPPN sequence generator at Alice is shown in Fig. \ref{fg:DSSS_simu}. Here, $c_{2,\rm max}$ is set to $4.88\times 10^{-7}$, and SNR$\rm_p = 30$ dB.
					 With estimated CSI, the BER performance of SE-AFDM with spreading factor $F = 0$ deteriorates at Bob due to incorrect received state parameters, which indicates the LPPN sequence synchronization failure.
					 In contrast, with spreading factor $F = 15$, the BER performance  at Bob remains consistent with that in Fig. \ref{fg:var_doppler}. The BER performance at Eve remains unchanged due to the inability to synchronize with Alice.

					\vspace{-2ex}
					
					\subsection{Experimental Results}
					
					 

				    To validate the feasibility of the proposed SE-AFDM communication system and synchronization framework, we conducted an over-the-air experiment under a two-path propagation scenario, as illustrated in Fig. \ref{fg:Simplified platform}. As shown in Fig. \ref{fg:Scenario}, the experimental setup consists of a transmitter and a receiver.			    
				    The transmitter employs a $21$ dBi omnidirectional antenna (Tx) connected to an SDR platform for signal generation and transmission. The receiver is configured with two omnidirectional antennas, Rx$1$ with a gain of $12$ dBi and Rx$2$ with a gain of $21$ dBi, both linked to the SDR platform. 				    
					The signals received by the two receiving antennas are fed into the SDR platform via a combiner.
					Tx is $2.473$ meters from Rx$1$ and $20.268$ meters from Rx$2$. With both receiving antennas connected to the SDR platform via $10$-meter SMA cables, the path delay arises solely from the disparity in free-space propagation distances. 
					Moreover, we simulate a frequency offset of $22.222$ kHz, corresponding to a velocity of $4800$ km/h.
					The detailed experimental parameters are provided in Table \ref{tab:table3}, and the results are presented in Fig. \ref{fg:kongkui}.
					
				
					 

					
					\vspace{-10pt}
				    \renewcommand{\arraystretch}{1.1}  
				    \begin{table}[h]
				    	\caption{Experimental parameters\label{tab:table3}}  
				    	\centering
				    	\begin{tabular}{|c|m{5cm}|c|}
				    		\hline 
				    		\textbf{Symbol} &\textbf{Parameter}  &\textbf{Value}\\ 
				    		\hline 
				    		$f_c$ &Carrier frequency &5 GHz\\
				    		\hline
				    		$B$ &Bandwidth &49.152 MHz\\
				    		\hline
				    		$\Delta f$ &Subcarrier spacing &48 kHz\\
				    		\hline
				    		$N$ &Number of subcarriers &1024\\
				    		\hline 
				    		$N_{\rm cp}$ &Number of CPP &13\\
				    		\hline 
				    		$M$ &Size of the codebook &1024\\
				    		\hline
				    		$N_{\rm sym}$ &Number of AFDM symbols per frame &256\\
				    		\hline 	
				    	\end{tabular}
				    \end{table}
					
					\begin{figure}[htbp] 
						\vspace{-5pt}
						\centering
						\includegraphics[width=3in]{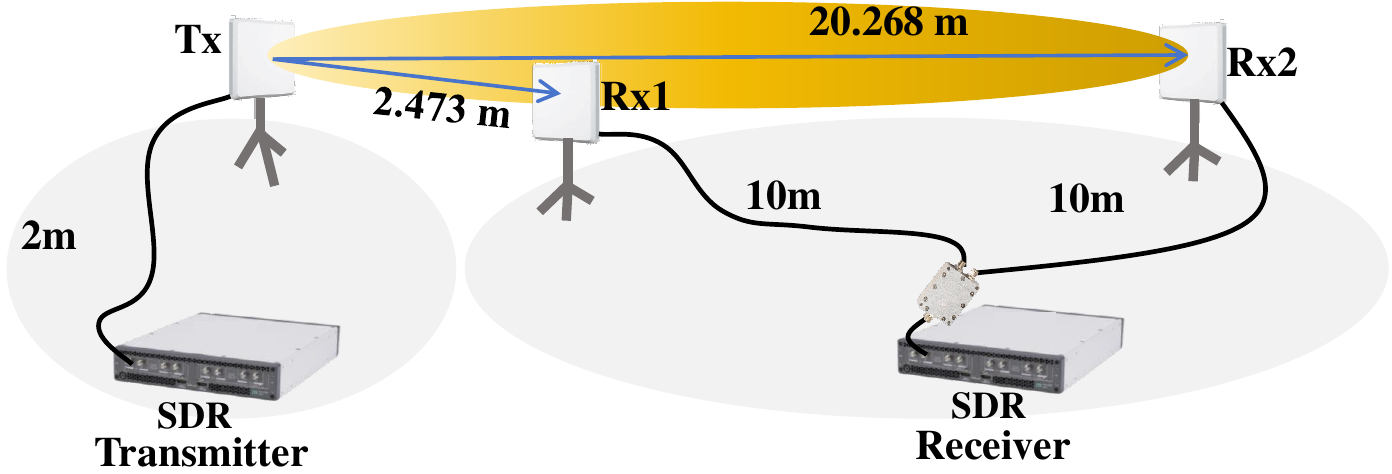}
						\caption{Schematic diagram of the experimental platform.
								\label{fg:Simplified platform}}
							\vspace{-5pt}
						\end{figure}

					
					Fig. \ref{fg:kongkui} shows the measured BER performances with $c_{2,\rm max} = 4.88\times 10^{-7}$ for Bob and Eve.
					Bob utilizes the state parameters with spreading factor $F=15$ of the LPPN sequence generator at Alice to synchronize the LPPN sequences between Alice and Bob. After LPPN sequence synchronization, the synchronized $\mathbf{c}^{\rm B}_{2,\mu}$ is generated for demodulation. 
					The BER at Bob decreases as the SNR increases, 
					exhibiting a trend consistent with the simulated results, i.e., the blue curves in Fig. \ref{fg:var_doppler} and the green curve in Fig. \ref{fg:DSSS_simu}.
					In contrast, the BER of Eve exhibits an error floor around $0.1$ as SNR increases, since Eve cannot synchronize with the varying $\mathbf{c}^{\rm A}_{2,\mu}$ of Alice. In summary, Fig. \ref{fg:kongkui} validates the effectiveness of the proposed synchronization framework in a real propagation environment.
					
					
						\begin{figure}[htbp]
						\vspace{-5pt}
						\centering
						\includegraphics[width=3.5in]{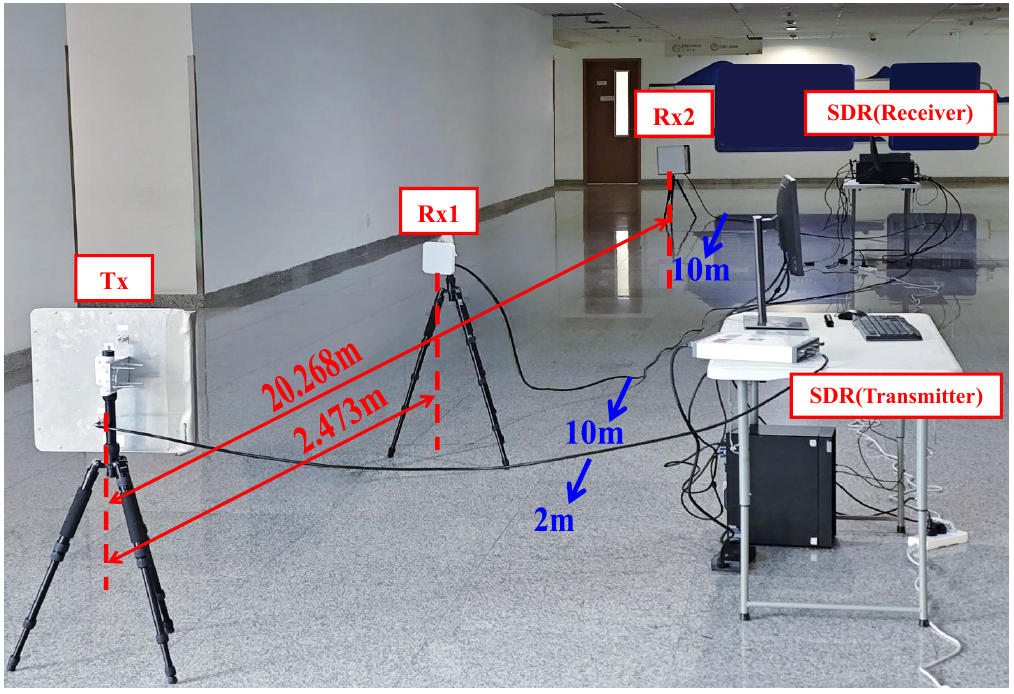}
						\caption{Over-the-air test scenario of the experimental platform.
								\label{fg:Scenario}}
							\vspace{-5pt}
						\end{figure}

					\begin{figure}[H]
						\vspace{-5pt}
						\centering
						\includegraphics[width=2.5in]{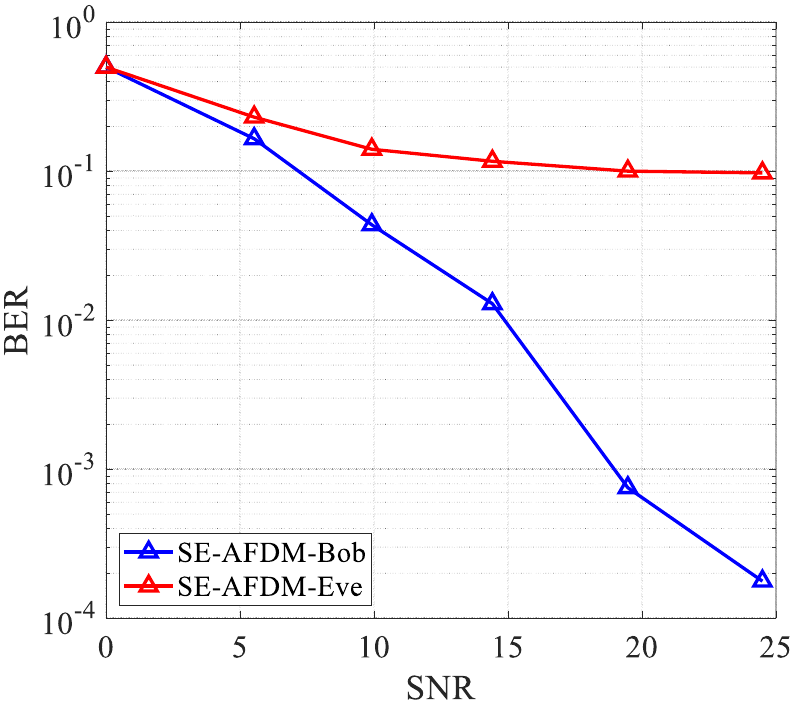}
						\caption{The experimental BER performances versus SNR of the SE-AFDM system with spreading factor $F=15$.
								\label{fg:kongkui}}
							\vspace{-5pt}
						\end{figure}
					

						\section{Conclusion} \label{sec:conclusion}
						
						This paper introduced an SE-AFDM communication system to enhance communication security by using the time-varying parameter $c_2$. 
						The varying $c_2$ was dynamically adjusted using an LPPN sequence and synchronized according to the proposed synchronization strategy between the legitimate receiver and the transmitter.
						Theoretical analysis showed that our SE-AFDM system can significantly improve communication security by configuring an appropriate parameter $c_2$.
						Numerical results revealed that Bob experiences no BER performance degradation compared with the conventional AFDM system, while Eve is unable to eavesdrop on information from Alice. 
						Experimental results verified the effectiveness of our proposed synchronization strategy.

	\vspace{-2ex}
	{
	\appendix  
		\vspace{-1ex}
	\section*{Proof of proposition \ref{proposition1}}   \label{proofcorollary1}
	\vspace*{-1pt}
     The effective output SINR of the $q$-th symbol at Eve, originally given in (\ref{eq:Eq25_1}), can be rewritten as:
	 \begin{align}
	 	\label{eq:SINR1}
	 	&{{\mathop{\rm SINR}\nolimits} _{E,q}} \!= \!\frac{{ \mathbb{E}\left\{ {|x[q]{|^2}} \right\}}}{\mathbb{E}{\left\{ {|x[q]{|^2}} \right\}\mathbb{E}\left\{ {{{\left| {{e^{j2\pi c_{2,\mu}^{\rm A}[q]{q^2}}} \!-\! 1} \right|}^2}} \right\}\! +\! \sigma _{n,E}^2}}.
	 \end{align} \normalsize
	 
	The transmit power of Alice is ${p_s}$ and the large-scale fading from Alice to Eve is $\alpha_E$. Then, (\ref{eq:SINR1}) can be reformulated as	
	 \begin{align} 	\label{eq:SINR1_1}
	 &{{\mathop{\rm SINR}\nolimits} _{E,q}} = \frac{{{p_s}\alpha _E^2}}{{{p_s}\alpha _E^2{\rm{ }} \mathbb{E}\left\{ {{{\left| {{e^{j2\pi c_{2,\mu}^{\rm A}[q]{q^2}}} - 1} \right|}^2}} \right\}+ \!\!\sigma _{n,E}^2}}\nonumber\\ 
	 &=\!\frac{\gamma_E}{\gamma_E \left\{2-\!\!\left(\mathbb{E}\left(e^{j 2 \pi c_{{2,\mu}}^{{\rm A}}[q] q^{2}}\right)\!\!+\!\!\mathbb{E}\left(e^{-j 2 \pi c_{{2,\mu}}^{{\rm A}}[q] q^{2}}\right)\right)\right\}\!\!+\!\!1},
	\end{align} \normalsize
	where $\gamma_E = {p_s} \alpha ^2_E/{{\sigma _{n,E}^2}}$ denotes the output SNR of the received signal at Eve. 
	The elements of the codebook $\mathcal{C}_2$ are specified in \eqref{eq:Az}, where the codebook interval ${\Delta_{\rm A}}=\frac{{2{c_{2,\max }}}}{M-1}$.
	The value ${c}^{\rm A}_{2,\mu}[q]$ is randomly drawn from the codebook $\mathcal{C}_2$ via an index generated by a truncated LPPN. Hence, we have
	\begin{align}
		\label{eq:SINR2}
			{\mathbb{E}\left(e^{j 2 \pi c_{{2,\mu}}^{{\rm A}}[q] q^{2}}\right)} &= \frac{1}{M}\sum\limits_{k = 0}^{M - 1} {{e^{j2\pi {q^2}\left( { - {c_{2,\max }} + k{\Delta_{\rm A}}  } \right)}}}.
	\end{align} \normalsize	
	If ${\Delta_{\rm A}}{q^2}   \in \mathbb{Z}$,	${\mathbb{E}\left(e^{j 2 \pi c_{{2,\mu}}^{{\rm A}}[q] q^{2}}\right)}$ = ${\mathbb{E}\left(e^{-j 2 \pi c_{{2,\mu}}^{{\rm A}}[q] q^{2}}\right)}$ = $1$. In other cases, \eqref{eq:SINR2} can be written as
\begin{eqnarray} \label{eq:SINR3}
	\mathbb{E}\left( {{e^{j2\pi c_{2,\mu}^{\rm A}[q]{q^2}}}} \right) = \frac{1}{M}{e^{ - j2\pi {q^2}{c_{2,\max }}}}\frac{{ {e^{j2\pi {\Delta _{\rm{A}}}{q^2}M}} -1}}{{ {e^{j2\pi {\Delta _{\rm{A}}}{q^2}}} -1 }}
\end{eqnarray}	

Similarly, we obtain
\begin{eqnarray} 	\label{eq:SINR4}
	\mathbb{E}\left( {{e^{-j2\pi c_{2,\mu}^{\rm A}[q]{q^2}}}} \right) = \frac{1}{M}{e^{  j2\pi {q^2}{c_{2,\max }}}}\frac{{ {e^{-j2\pi {\Delta _{\rm{A}}}{q^2}M}}-1  }}{{{e^{-j2\pi {\Delta _{\rm{A}}}{q^2}}} -1}}
\end{eqnarray}	
Substituting \eqref{eq:SINR3} and \eqref{eq:SINR4} into \eqref{eq:SINR1_1}, the final expression is given by
\begin{align}\label{eq:SINR_last}
&{{\mathop{\rm SINR}\nolimits} _{{\rm{E}},q}} =\nonumber\\ 
& \left\{ {\begin{array}{*{20}{l}}
		\hspace{-2ex} {{\gamma _{\rm{E}}},}&\hspace{-1.8ex}{{\Delta _{\rm{A}}}{q^2} \in \mathbb{Z},}\\
     	\hspace{-2ex} {\frac{{{\gamma _{\rm{E}}}}}{{{\gamma _{\rm{E}}}\!\left\{ {\!2 -\! \frac{2}{M}\!\Re \left\{\!\! {{e\!^{ - j2\pi {c_{2,\max }}{q^2}}}\!\frac{{{e^{j2\pi {\Delta _{\rm{A}}}{q^2}M}}\! -\! 1}}{{{e^{j2\pi {\Delta _{\rm{A}}}{q^2}}} \!- \!1}}\! } \right\} } \!\!\right\}\! +\! 1}}}\!,&\hspace{-1.8ex}{{\text{otherwise}}{\rm{.}}}
\end{array}} \right.
\end{align}

}

	\bibliographystyle{IEEEtran}
	\bibliography{IEEEabrv,refer}

%
%
%
%
%

	\vfill
	
	\end{document}